\documentclass[11pt, a4paper]{amsart}  
\usepackage{amsmath,amssymb,amsfonts,times,hyperref}

\title[SDP, harmonic analysis and coding theory]{Semidefinite programming, harmonic analysis and coding theory}

\author{Christine Bachoc} 
\address{C. Bachoc, Institut de Math\'ematiques de Bordeaux, Universit\'e Bordeaux I, 351,
cours de la Li\-b\'e\-ration, 33405 Talence France}
\email{bachoc@math.u-bordeaux1.fr}


\date{today}
\newtheorem{defi}{Definition}[section]
\newtheorem{definition}[defi]{Definition}
\newtheorem{proposition}[defi]{Proposition}
\newtheorem{theorem}[defi]{Theorem}
\newtheorem{remark}[defi]{Remark}

\newtheorem{lemma}[defi]{Lemma}

\newcommand{\Z}{{\mathbb{Z}}} 

\newcommand{\R}{{\mathbb{R}}} 
\newcommand{\C}{{\mathbb{C}}}
\newcommand{\F}{{\mathbb{F}}} 
\newcommand{\OO}{{\mathbb{O}}} 
\newcommand{\Pp}{{\mathbb{P}}} 

\newcommand{\CC}{{\mathcal{C}}} 
\newcommand{\PP}{{\mathcal{P}}} 
\newcommand{\RR}{{\mathcal{R}}} 
\newcommand{\K}{{\mathcal{K}}}

\newcommand{\I}{{\mathcal{I}}} 
\newcommand{\Ga}{{\Gamma}} 
\newcommand{\Sn}{S^{n-1}} 
\newcommand{\Snm}{S^{n-2}} 
\newcommand{\Pold}{\operatorname{Pol}_{\leq d}}
\newcommand{\Polk}{\operatorname{Pol}_{\leq k}}
\newcommand{\Polki}{\operatorname{Pol}_{\leq k+i}}
\newcommand{\Pk}{P_{\kappa}}
\newcommand{\On}{{\operatorname{O}(\R^n)}}

\newcommand{\Aut}{\operatorname{Aut}}
\newcommand{\End}{\operatorname{End}}
\newcommand{\Trace}{\operatorname{Trace}}
\newcommand{\shape}{\operatorname{shape}}
\newcommand{\Harm}{\operatorname{Harm}}
\newcommand{\1}{\operatorname{\bf 1}}

\newcommand{\card}{\operatorname{card}}
\newcommand{\Stab}{\operatorname{Stab}}
\newcommand{\Res}{\operatorname{Res}}
\newcommand{\Ind}{\operatorname{Ind}}
\newcommand{\im}{\operatorname{im}}

\newcommand{\Gl}{\operatorname{Gl}}
\newcommand{\Id}{\operatorname{Id}}
 
\newcommand{\Gmn}{{\mathcal{G}_{m,n}}}
\newcommand{\qbinom}[2]{\genfrac{[}{]}{0pt}{}{#1}{#2}}

\begin{document}

\begin{abstract} 
These lecture notes where presented as a course of the CIMPA summer
school in Manila, July 20-30,  2009, {\it Semidefinite programming in
  algebraic combinatorics}. This version is an update of June 2010.
\end{abstract}

\maketitle

\tableofcontents

\section{Introduction}

In coding theory, the so-called linear programming method, introduced 
by Phi\-lip\-pe Delsarte in the seventies \cite{Del1} as proved to be a
very powerful method to solve extremal problems. It was initially
developed in the framework of association schemes and then extended 
to the family of $2$-point homogeneous spaces, including the compact
real manifolds having this property (see \cite{DGS},
\cite{KL}, \cite[Chapter 9]{CS}). Let us recall that a $2$-point
homogeneous space is a  metric space on which a group  $G$ acts
transitively,
leaving the distance $d$ invariant, and such that, for $(x,y)\in X^2$,
there exists $g\in G$ such that $(gx,gy)=(x',y')$ if and only if $d(x,y)=d(x',y')$.
The Hamming space $H_n$ and the unit sphere of the Euclidean space
$S^{n-1}$
are core examples of such spaces which play a major role in coding theory. To such a space
is associated a sequence of orthogonal polynomials $(P_k)_{k\geq 0}$
such that, for all $C\subset X$,
\begin{equation*}
\sum_{(c,c')\in C^2} P_k(d(c,c'))\geq 0.
\end{equation*}
These inequalities can be understood as linear constraints on the
distance distribution of a code and are at the heart of the LP method.

 The applications of
this method  to the study of codes and designs are numerous:
very good upper bounds for the number of elements of a code with given
minimal distance can be obtained with this method, including a number
of
cases where this upper bound is tight and leads to a proof of
optimality and uniqueness of certain codes, as well as to the best known 
asymptotic bounds (see \cite{Del1}, \cite{MRRW}, \cite{KL}, \cite[Chapter 9]{CS}, \cite{L}).

In recent years, the development of the theory of error correcting
codes has introduced many other spaces with interesting applications.
To cite a few, 
codes over various alphabets associated to various weights, quantum
codes,
codes for the multi antenna systems of communications involving more 
complicated manifolds like  the Grassmann spaces, have successively
focused attention. For these spaces there was a need for a
generalization of the classical framework of the linear programming
method. This generalization was developed for some of these spaces,
see \cite{TAG}, \cite{T}, \cite{B1}, \cite{Roy2}. 
It turns out  that in each of these cases, a certain sequence of orthogonal polynomials enters into play
but unlike the classical cases, these polynomials are multivariate.

Another step was taken when A. Schrijver in \cite{Schrijver} succeeded to improve the
classical LP bounds for binary codes with the help of semidefinite
programming. To that end he exploited {\em SDP constraints on triples
  of points} rather than on pairs, arising from the analysis of the
Terwilliger algebra of the Hamming scheme.
His method was then adapted to the unit sphere \cite{BV1} in the
framework of the representations of the orthogonal group. The heart of
the method is to evidence 
matrices $Z_k(x,y,z)$  such that 
for all $C\subset X$,
\begin{equation*}
\sum_{(c,c',c'')\in C^3} Z_k(c,c',c'')\succeq  0.
\end{equation*}

Another motivation for the study of SDP constraints on $k$-tuples of
points can be found in coding theory. It appears that not only
functions on pairs of points such as a distance function $d(x,y)$ are
of interest, but also functions on $k$-tuples have relevant meaning,
e.g. in connection with the notion of  list decoding.

In these lecture notes we want to develop a general framework based on 
harmonic analysis of compact groups for these methods. In view of the
effective applications to coding theory, we give detailed computations
in many cases. Special attention will be paid to the cases of the
Hamming space and of the unit sphere. 

Section 2 develops the basic tools needed in the theory of
representations of finite groups, section 3 is concerned with the
representations of compact groups and Peter Weyl theorem. Section 4
discusses the needed notions of harmonic analysis: the zonal matrices
are introduced and the invariant positive definite functions are
characterized with Bochner theorem. Section 5 is devoted to explicit
computations of the zonal matrices. Section 6 shows how the
determination of the invariant positive definite functions leads to an
upper bound for codes with given minimal distance. Section 7 explains the  connection with the
so-called Lov\'asz theta number. Section 8 shows how SDP bounds can be
used to strengthen the classical LP bounds, with the example of the
Hamming space.

\subsection{Notations:} for a matrix $A$ with complex coefficients,
$A^*$ stands for the transposed conjugate matrix. A squared matrix is
said to be hermitian if $A^*=A$ and positive semidefinite if it is
hermitian and all its eigenvalues are non negative. This property is
denoted $A\succeq 0$. 
We follow standard notations
for sets of matrices:
the set of $n\times m$ matrices with coefficients in a field $K$ is
denoted $K^{n\times m}$; the group of $n\times n$ invertible matrices by
$\Gl(K^n)$;  the group $U(\C^n)$ of unitary matrices, respectively
$O(\R^n)$ of
orthogonal matrices is the set of
matrices $A\in \Gl(C^{n})$, respectively $A\in \Gl(\R^{n})$ such that
$A^*=A^{-1}$.
The space $\C^{n\times m}$ is endowed with the standard inner product 
$\langle A, B\rangle=\Trace(AB^*)=\sum_{i,j} A_{i,j} \overline{B_{i,j}}$. 
The number of elements of a finite set $X$ is
denoted $\card(X)$ of $|X|$. 

\section{Linear representations of finite groups}

In this section we shortly review the basic notions of group
representation theory that will be needed later. There are many good
references for this theory e.g. \cite{Serre}, or  \cite{Sagan} which 
is mainly devoted to the symmetric group.

\subsection{Definitions}
Let $G$ be a finite group. A (complex linear) representation of
$G$
is a finite dimensional complex vector space $V$ together with a
homomorphism $\rho$:
\begin{equation*}
\rho: G \to \Gl(V)
\end{equation*}
where $\Gl(V)$ is the general linear group of $V$, i.e. the set of
linear invertible transformations of $V$. The degree of the
representation $(\rho,V)$ is by definition equal to the dimension of
$V$. 

Two representations of $G$ say $(\rho, V)$ and $(\rho',V')$ are said
to be equivalent or isomorphic if there exists and isomorphism $u: V\to V'$ such
that, for all $g\in G$,
\begin{equation*}
\rho'(g)=u\rho(g)u^{-1}.
\end{equation*}
For example, the choice of a basis of $V$ leads to a 
representation equivalent to $(\rho,V)$  given by $(\rho', \C^d)$ where $d=\dim(V)$ and
$\rho'(g)$
is the matrix of $\rho(g)$ in the chosen basis. In general, a
representation of $G$ such that $V=\C^d$ is called a matrix
representation.

The notion of a $G$-module is  equivalent to the above
notion of representation and turns out to be very convenient.
A $G$-module, or a $G$-space,
is a finite dimensional complex vector space $V$ such that for all
$g\in G$, $v\in V$, $gv\in V$ is well defined and satisfies the obvious
properties: $1v=v$, $g(hv)=(gh)v$,  $g(v+w)=gv+gw$, $g(\lambda
v)=\lambda(gv)$
for $g,h\in G$, $v,w\in V$, $\lambda\in \C$. In other words, $V$ is
endowed with a structure of $\C[G]$-module.
One goes from one notion to the other by the identification
$gv=\rho(g)(v)$. The notion of equivalent representations corresponds
to the notion of isomorphic $G$-modules, an isomorphism of
$G$-modules
being an isomorphism of vector spaces $u:V\to V'$ such that
$u(gv)=gu(v)$.
Note that here the operations of $G$ on $V$ and $V'$ are denoted
alike,
which may cause some confusion.

\subsection{Examples}

\begin{itemize}
\item The trivial representation $\1$: $V=\C$ and $gv=v$.
\item Permutation representations: let $X$ be a finite set on which $G$
  acts (on the left).
Let $V_X:=\oplus_{x\in X} \C e_x$. A natural action of $G$ on $V_X$ is
given by $ge_x=e_{gx}$, and defines a representation of $G$, of
degree $|X|$. The matrices of this representation (in the basis
$\{e_x\}$)
are permutation matrices.
\begin{itemize}
\item The symmetric group $S_n$ acts on $X=\{1,2,\dots,n\}$. This
  action defines a representation of degree $n$ of $S_n$.
\item For all $w$, $1\leq w\leq n$, $S_n$ acts on the set $X_w$ of subsets
  of $\{1,2,\dots,n\}$ of cardinal $w$. In coding theory an element of
  $X_w$
is more likely viewed as a binary word of length $n$ and Hamming
weight $w$. The spaces $X_w$ are called the Johnson spaces and denoted
$J_n^w$.
\end{itemize}
\item The regular representation is obtained with the special case 
$X=G$ with the action of $G$ by left multiplication. 
In the case $G=S_n$ it has degree $n!$.. It turns out that the 
regular representation contains all building blocks of all
representations of $G$.
\item Permutation representations again: if $G$ acts transitively on
  $X$,
this action can be identified with the left action of $G$ on the left
cosets
$G/H=\{gH: g \in G\}$ where $H=\Stab(x_0)$ is the stabilizer of a
base point.
\begin{itemize}
\item The symmetric group $S_n$ acts transitively on
  $X=\{1,2,\dots,n\}$
and the stabilizer of one point (say $n$) can be identified with the symmetric group
$S_{n-1}$ acting on 
$\{1,\dots,n-1\}$.
\item The action of $S_n$ on $J_n^w$ is also transitive and the
  stabilizer of one point (say $1^w0^{n-w}$) is the subgroup
$S_{\{1,\dots,w\}}\times S_{\{w+1,\dots,n\}}$ isomorphic to
  $S_{w}\times S_{n-w}$.
\item The Hamming space $H_n=\{0,1\}^n=\F_2^n$ affords the transitive action 
of $G=T\rtimes S_n$ where $T$ is the group of translations 
$T=\{t_u: u\in H_n\}$, $t_u(v)=u+v$ and $S_n$ permutes the
coordinates. 
The stabilizer of $0^n$ is the group of permutations  $S_n$. 
\end{itemize}
\item Another way to see the permutation representations is the following: let 
\begin{equation*}
\CC(X):=\{f:X\to\C\}
\end{equation*}
 be the space of functions from $X$ to $\C$. 
The action of $G$ on $X$ extends to a structure of $G$-module on 
$\CC(X)$ given by:
\begin{equation*}
gf(x):=f(g^{-1}x).
\end{equation*}
For the Dirac functions $\delta_y$ ($\delta_y(x)=1$ if $x=y$, $0$
otherwise), the action of $G$ is given by $g\delta_y=\delta_{gy}$ 
thus this representation is isomorphic to the permutation
representation defined by $X$. This apparently more complicated presentation of
permutation representations has the advantage to allow generalization to infinite
groups acting on infinite spaces as we shall encounter later.

\end{itemize}

\subsection{Irreducibility}

Let $V$ be a $G$-module (respectively a representation $(\rho,V)$ of
$G$). A subspace $W\subset V$ is said to be $G$-invariant (or
$G$-stable, or a $G$-submodule, or a subrepresentation of
$(\rho,V)$),
if $gw\in W$ (respectively $\rho(g)(w)\in
W$) for all $g\in G$, $w\in W$ . 

\paragraph{\bf Example:} $V=V_{G}$ and $W=\C e_{G}$ with
$e_{G}=\sum_{g\in G} e_g$. The restriction of the action of $G$
to $W$ is the trivial representation.

A $G$-module $V$  is said to
be irreducible if it does not contain any subspace $W$, $W\neq \{0\},
V$, invariant under $G$. Otherwise it is called reducible.

\paragraph{\bf Example:} The $G$-invariant subspaces of dimension $1$
are necessarily irreducible. If $G$ is abelian, a $G$-module of
dimension greater than $1$ cannot be irreducible, because
endomorphisms that pairwise commute afford a common basis of
eigenvectors.

\smallskip
The main result is then the decomposition of a $G$-module into the
direct sum of irreducible submodules:

\begin{theorem}[Maschke's theorem]\label{t1}
Any $G$-module $V\neq \{0\}$ is the direct sum of irreducible
$G$-submodules $W_1,\dots, W_k$:
\begin{equation}\label{Maschke}
V=W_1\oplus W_2\oplus\dots\oplus W_k.
\end{equation}
\end{theorem}

\begin{proof} By induction, it is enough to prove that any
  $G$-submodule $W$ of $V$ affords a supplementary subspace which is
  also $G$-invariant. The main idea is to construct a $G$-invariant
  inner product and then prove that the orthogonal of $W$ for this
  inner product makes the job. 

We start with an inner product $\langle x,y\rangle$ defined on
$V$. There are plenty of them since $V$ is a finite dimensional
complex vector space. For example we can choose an arbitrary basis of
$V$ and declare it to be orthonormal. Then we average this inner
product on $G$, defining:
\begin{equation*}
\langle x,y\rangle':=\sum_{g\in G} \langle gx,  gy \rangle.
\end{equation*}
It is not difficult to check that we have defined a inner product
which is $G$-invariant. It is also easy to see that 
\begin{equation*}
W^{\perp}:= \{v\in V : \langle v,w\rangle'=0\text{ for all }w\in W\}
\end{equation*}
is $G$-invariant, thus we have the decomposition of $G$-modules:
\begin{equation*}
V=W\oplus W^{\perp}
\end{equation*}
\end{proof}

It is worth to notice that the above decomposition may not be
unique. It is clear if one thinks of the extreme case $G=\{1\}$ for
which the irreducible subspaces are simply the one dimensional
subspaces of $V$. The decomposition of $V$ into the direct sum of 
subspaces of dimension $1$ is certainly not unique (if $\dim(V)>1$ of
course).
But uniqueness is fully satisfied by the decomposition into isotypic
subspaces.
In order to define them we take the following notation: let 
$\RR$ be a complete set of pairwise non isomorphic
irreducible representations of $G$. We have seen that any $G$-module
affords a $G$-invariant inner product so the action of $G$ on $R$ is
expressed by unitary matrices in a given orthonormal matrix of
$R$. According to the context we view $R$ either as a $G$-module or as
a homomorphism $g\mapsto R(g)\in U(\C^n)$. It will turn out that there
is only a finite number of them but we have not proved it yet. 
The isotypic subspace $\I_R$ of $V$  associated to $R\in \RR$ is
defined, with the notations of  \eqref{Maschke}, by:
\begin{equation}\label{IR}
\I_R:=\oplus_{W_i\simeq R} W_i.
\end{equation}

\begin{theorem}\label{Isotypic}
Let $R\in \RR$. The isotypic spaces $\I_R$ do not depend on the decomposition of $V$
as the direct sum of $G$-irreducible subspaces. We have the canonical
decomposition
\begin{equation*}
V=\oplus_{R\in \RR} \I_R.
\end{equation*}
Any $G$-subspace $W\subset V$ such that $W\simeq R$ is
contained in $\I_R$ and any $G$-irreducible subspace of $\I_R$ is isomorphic to $R$.
A decomposition into irreducible subspaces of $\I_R$ has the form
\begin{equation*}
\I_R=W_1\oplus\dots \oplus W_{m_R}
\end{equation*}
with $W_i\simeq R$. Such a  decomposition is not unique in general but the number
$m_R$ does not depend on the decomposition and is called the
multiplicity of $R$ in $V$.

Moreover, if $V$ is endowed with a $G$-invariant inner product, then
the isotypic spaces are pairwise orthogonal.
\end{theorem}

\begin{proof}
We start with a lemma which points out a very important property of
irreducible $G$-modules.

\begin{lemma}[Schur Lemma]\label{Schur}
Let $R_1$ and $R_2$ two irreducible $G$-modules and  let 
$\varphi : R_1\to R_2$ be a $G$-homomorphism. Then either $\varphi=0$ or $\varphi$ is an
isomorphism of $G$-modules. 
\end{lemma}
\begin{proof} The subspaces $\ker \varphi$ and $\im \varphi$ are $G$-submodules
  of respectively $R_1$ and $R_2$ thus they are equal to either
  $\{0\}$ or $R_i$.
\end{proof}

We go back to the proof of the theorem. We start with the
decomposition  \eqref{Maschke}
of $V$ and the definition \eqref{IR} of $\I_R$, a priori depending on
the decomposition. Let $W\subset V$, a $G$-submodule isomorphic to
$R$.
We apply Lemma \ref{Schur} to the projections $p_{W_i}$ and conclude
that either $p_{W_i}(W)=\{0\}$ or $p_{W_i}(W)=W_i$ and this last case
can only happen if $W\simeq W_i$. It proves that $W\subset \I_R$ and
that a $G$-irreducible subspace of $\I_R$ can only be isomorphic to $R$. 
It also proves that 
\begin{equation*}
\I_R=\sum_{W\subset V, W\simeq R} W
\end{equation*}
hence giving a characterization of $\I_R$ independent of the initial
decomposition.
The number $m_R$ must satisfy $\dim(\I_R)=m_R\dim(R)$ so it is
independent
of the decomposition of $\I_R$.

If $V$ is equipped with a $G$-invariant inner product, we consider
orthogonal projections. Schur Lemma shows that $P_W(W')=\{0\}$ or $=W$
if $W$ and $W'$ are irreducible. Thus if they are not $G$-isomorphic,
$W$ and $W'$ must be orthogonal.

\end{proof}

\subsection{The algebra of $G$-endomorphisms} Let $V$ be a
$G$-module. The set of $G$-endomorphisms of $V$ is an  algebra
(for the laws of addition and composition)
denoted $\End_{G}(V)$. The next theorem describes the structure of
this algebra.

\begin{theorem}\label{End}
If $V\simeq \oplus_{R\in \RR} R^{m_R}$, then
\begin{equation*}
\End_{G}(V)\simeq \prod_{R\in \RR} \C^{m_R\times m_R}.
\end{equation*}
\end{theorem}

\begin{proof}
The proof is in three steps: we shall assume first $V=R$ is
irreducible, then $V\simeq R^m$, then the general case. Schur Lemma \ref{Schur} is
the main tool here.

If $V$ is irreducible, let $\varphi\in \End_{G}(V)$. Since $V$ is a
complex vector space, $\varphi$ has got an eigenvalue $\lambda$. Then
$\varphi-\lambda\Id$
is a $G$-endomorphism with a non trivial kernel so from Schur Lemma
$\varphi-\lambda\Id=0$. We have proved that 
\begin{equation*}
\End_{G}(V)=\{\lambda\Id, \lambda\in \C\}\simeq \C.
\end{equation*}
We assume now that $V\simeq R^m$ and we fix a decomposition
$V=W_1\oplus\dots \oplus W_m$. For all $1\leq i\leq j\leq m$, let
$u_{j,i}:W_i\to W_j$ an isomorphism of $G$-modules such that the
relations
\begin{equation*}
u_{k,j}\circ u_{j,i}=u_{k,i}\text{ and } u_{i,i}=\Id
\end{equation*}
hold for all $i,j,k$. Let $\varphi\in \End_{G}(V)$; we associate to
$\varphi$ an element of
$\C^{m\times m}$ in the following way. From previous discussion of the
irreducible case it follows that for all $i,j$ there exists
$a_{i,j}\in \C$ such that, for all $v\in W_i$,
\begin{equation*}
p_{W_j}\circ \varphi(v)=a_{j,i}u_{j,i}(v).
\end{equation*}
The matrix $A=(a_{i,j})$ is the matrix associated to $\varphi$.
The proof that the mapping $\varphi\mapsto A$ is an
isomorphism of algebras carries without difficulties and is left to the reader.

In the general case, $V=\oplus_{R\in \RR} \I_R$. Let $\varphi\in 
\End_{G}(V)$. It is clear that $\varphi(\I_R)\subset \I_R$ thus
\begin{equation*}
\End_{G}(V)=\oplus_{R\in \RR} \End_{G}(\I_R)
\end{equation*}
and we are done.
\end{proof}

It is worth to notice that $\End_{G}(V)$ is a commutative algebra if
and only if all the multiplicities $m_R$ are equal to either $0$ or
$1$.
In this case we say that $V$ is multiplicity free. It is also the
unique case when the decomposition into irreducible subspaces
\eqref{Maschke} is unique.

\subsection{Orthogonality relations}
Another important result which is a consequence  of Schur lemma is the orthogonality
relations between the matrix coefficients of the elements of $\RR$:

\begin{theorem}\label{orthog rel}
For $R\in \RR$, let $d_R:=\dim(R)$. For all $R,S\in \RR$, $i,j,k,l$, 
\begin{equation*}
\langle R_{i,j}, S_{k,l}\rangle =
\frac{1}{d_R}\delta_{R,S}\delta_{i,k}\delta_{j,l}.
\end{equation*}
\end{theorem}

\begin{proof} 
For $A\in \C^{d_R\times d_S}$, let
\begin{equation*}
A'=\frac{1}{|G|}\sum_{g\in G} R(g) A S(g)^{-1}.
\end{equation*}
This matrix satisfies $R(g)A'=A'S(g)$ for all $g\in G$.
In other words it defines an homomorphism of $G$-modules from 
$(\C^{d_S}, S)$ to $(\C^{d_R}, R)$. Schur lemma shows that 
if $S\neq R$, $A'=0$ and if $S=R$, $A'=\lambda \Id$. Computing the
trace of $A'$ shows that $\lambda=\Trace(A)/{d_R}$. Taking $A=E_{i,j}$
the elementary matrices,  with the property that $S(g)^{-1}=S(g)^*$,
leads to the announced formula.
\end{proof}

\subsection{Characters}
The character of a representation $(\rho,V)$ of $G$ is the function 
$\chi_{\rho}:G\to \C$ defined by
\begin{equation*}
\chi_{\rho}(g)=\Trace(\rho(g)).
\end{equation*}
As a consequence of the standard property of traces of matrices
$\Trace(AB)=\Trace(BA)$,
the character of a representation only depends on its equivalence
class, and  it is a complex valued function on $G$ which
is constant on the conjugacy classes of $G$ (such a function is
called a class function). The inner product of any two $\chi$, $\psi\in \CC(G)$
is defined by
\begin{equation*}
\langle \chi,\psi\rangle:=\frac{1}{|G|}\sum_{g\in G}
\chi(g)\overline{\psi(g)}.
\end{equation*}
We have the very important orthogonality relations between characters:

\begin{theorem}[Orthogonality relations of the first kind]\label{orth rel}
Let $\chi$ and $\chi'$ be respectively the characters of two
irreducible representations $(\rho,V)$ and $(\rho',V')$ of $G$. Then
\begin{equation*}
\langle \chi,\chi'\rangle=\begin{cases}
1 \quad \text{ if }\rho\simeq \rho'\\
0 \quad \text{ otherwise.}
\end{cases}
\end{equation*}
\end{theorem}

\begin{proof} It is a straight forward consequence of Theorem
  \ref{orthog rel}, since the trace of a representation is the sum of the diagonal
  elements of any equivalent matrix representation.
\end{proof}

A straightforward consequence of the above theorem is that 
$\langle \chi_{\rho}, \chi_R\rangle=m_R$ for all $R\in \RR$.
This property is a very convenient tool to study the irreducible
decomposition
of a given representation $(\rho,V)$ of $G$; in particular it
shows that a representation is irreducible if and only if its
character $\chi$ satisfies $\langle \chi,\chi\rangle=1$.
In the case 
of the regular representation it leads to the following very important result:

\begin{theorem}\label{Dec Reg Rep}[Decomposition of the regular representation]
\begin{equation*}
\CC(G)\simeq \oplus_{R\in \RR} R^{\dim(R)}
\end{equation*}
\end{theorem}
\begin{proof} Compute the character of the regular representation.
\end{proof}
A consequence of the above theorem is the finiteness of the
number of irreducible representations of a given finite group,
together
with the formula
\begin{equation*}
|G|=\sum_{R\in \RR} (\dim(R))^2
\end{equation*}
which shows e.g. completeness of a given set of
irreducible $G$-modules.

A second consequence of the orthogonality relations is that a representation 
of $G$ is uniquely characterized up to isomorphism by its character.
\begin{theorem}
\begin{equation*}
(\rho,V)\simeq (\rho',V') \Longleftrightarrow
  \chi_{\rho}=\chi_{\rho'}.
\end{equation*}
\end{theorem}
\begin{proof}
If $\chi_{\rho}=\chi_{\rho'}$, the multiplicities of an irreducible
representation
of $G$ are the same in $V$ and $V'$, hence $V\simeq_{G} V'$.
\end{proof}

Let us denote by $R(G)$ the subspace of elements of $\CC(G)$ which are
constant on the conjugacy classes $C_1,\dots, C_s$ of $G$.  The
dimension of $R(G)$ is obviously the number $s$ of conjugacy classes
of $G$. We have seen that the characters $\chi_R$ of the irreducible
representations of $G$ belong to $R(G)$ and form an orthonormal
family. It turns out that they in fact form a basis of $R(G)$, which
in other words means that the number of irreducible representations of
$G$ is exactly equal to its number of conjugacy classes.

\begin{theorem}
\begin{equation*}
R(G)=\oplus_{R\in \RR} \C \chi_R.
\end{equation*}
\end{theorem}

\begin{proof}
It is clear that $\CC(G)=\C[G]\delta_e$. Thus $\End_G(\CC(G))\simeq
\C[G]$. In particular, the center of $\End_G(\CC(G))$ is isomorphic to
the center $Z(\C[G])$ of $\C[G]$. It is easy to verify that the center of $\C[G]$
is the vector space spanned by the elements $\lambda_i:=\sum_{g\in
  C_i} g$ associated to each conjugacy class $C_i$ of $G$, thus
$Z(\C[G])$ is of dimension $s$ the number of conjugacy classes of $G$.
On the other hand, as a consequence of Theorem \ref{Dec Reg Rep} and
Theorem \ref{End},
we have $\End_G(\CC(G))\simeq \prod_{R\in \RR} \C^{d_R\times d_R}$
where $d_R=\dim(R)$. Thus the center of $\End_G(\CC(G))$ is isomorphic
to $\C^{|\RR|}$ and we have proved that the number of $G$-irreducible
modules is equal to the number of conjugacy classes of $G$.
\end{proof}

\begin{remark} There is not in general a natural bijection between
  the set of conjugacy classes of $G$ and the set of its irreducible
  representations. However, in the special case of the symmetric group
  $S_n$, such a correspondance exists. The conjugacy classes are
  naturally indexed by the partitions $\lambda$ of $n$ and to every
  partition $\lambda$ of $n$ is associated an irreducible module
  $S^{\lambda}$ also called a Specht module (see \cite{Sagan}).
\end{remark}

\subsection{Induced representation and Frobenius reciprocity}
Induction is a way to construct representations of a group $G$ from
representations of its subgroups. Looking at the irreducible subspaces
of representations that are induced from subgroups is a   very convenient way to
find new irreducible representations of a group $G$. Induction is an
operation on representations which is dual to the easier to understand
restriction. If $V$ is a $G$-module and $H$ is a subgroup of $G$, the
restriction $\Res^G_H (V)$ is simply the space $V$ considered as a
$\C[H]$-module. If $V$ is an $H$-module, we define $\Ind_H^G(V)$ to be
the $\C[G]$-module 
\begin{equation*}
\Ind_H^G(V):= \C[G]\otimes_{\C[H]} V.
\end{equation*}
Here we
exploit the bi-module structure of $\C[G]$ (the tensor product over
$\C[H]$ means that $\lambda\mu\otimes v=\lambda\otimes \mu v$ when
$\mu\in \C[H]$). A more explicit (but less intrisic) formulation for
this construction is the following: let $\{x_1,\dots,x_t\}$ be a
complete system of representatives of $G/H$, so that $G=x_1H\cup\dots\cup x_t H$. Then 
\begin{equation*}
\Ind_H^G(V)=\oplus_{i=1}^t x_iV
\end{equation*}
where the left action of $G$ is as follows: for all $i$, there is $j$ and
$h\in H$ both depending on $g$ such that $gx_i=x_j h$. Then $gx_iv:=x_j(hv)$
where $hv\in V$. 
A third construction of $\Ind_H^G(V)$ is the following:
\begin{equation*}
\Ind_H^G(V)=\{f:G\to V \text{ such that } f(gh)=h^{-1}f(g)\}.
\end{equation*}
The equivalence of these three formulations is a recommended exercise !

\smallskip
\paragraph{\bf Example:} The permutation representation of $G$ on
$X=G/H$ is nothing else than the induction of the trivial
representation of $H$. In short, $\CC(X)=\Ind_H^G \1$.

\smallskip
Since the induction of two isomorphic $H$-modules leads to isomorphic
$G$-modules and similarly for the restriction, these
operations act on the characters thus we denote similarly $\Res_H^G
\chi$, $\Ind_H^G \chi$ the characters of the corresponding modules.  

\begin{lemma} Let $\chi$ be a character of $H$. The induced character $\Ind_H^G \chi$ is
  given by the formula:
\begin{equation*}
\Ind_H^G \chi (g)=\frac{1}{|H|} \sum_{\substack{ x\in    G\\x^{-1}gx\in H}} \chi(x^{-1}gx).
\end{equation*}
\end{lemma}
\begin{proof}
We take a decomposition $\Ind_H^G(V)=x_1V\oplus\dots \oplus x_tV$ where
$\{x_1,\dots,x_t\}$ are representatives of $G/H$. Since $gx_iv=x_jhv$
with the notations above, $g(x_iV)\subset x_jV$ and the block $x_iV$ will contribute to the
trace of $x\mapsto gx$ only when $j=i$, which corresponds to the case
when $x_i^{-1}gx_i\in H$. Then, the multiplication by $g$ on
$x_iV$ acts like the multiplication by $h=x_i^{-1}gx_i$ on V. Thus we
have
\begin{align*}
\Ind_H^G \chi (g)&=\sum_{\substack{ 1\leq i\leq t\\x_i^{-1}gx_i\in H}} \chi(x_i^{-1}gx_i)\\
&=\frac{1}{|H|} \sum_{\substack{ x\in    G\\x^{-1}gx\in H}}
\chi(x^{-1}gx).
\end{align*}
\end{proof}

The duality between the operations of restriction and induction is
expressed in the following important theorem:

\begin{theorem}[Frobenius reciprocity]
Let $H$ be a subgroup of $G$ and let $\chi$ and $\psi$ be respectively
a character of $H$ and a character of $G$. Then
\begin{equation*}
\langle \Ind_H^G \chi, \psi\rangle=\langle \chi, \Res_H^G \psi\rangle.
\end{equation*}
\end{theorem}

\begin{proof} Let $\tilde{\chi}:G\to \C$ be defined by:
  $\tilde{\chi}(g)=\chi(g)$ if $g\in H$ and   $\tilde{\chi}(g)=0$ if
  $g\notin H$ (of course $\tilde{\chi}$ is not
  a character of $G$). We compute $\langle \Ind_H^G \chi, \psi\rangle$:
\begin{align*}
\langle \Ind_H^G \chi, \psi\rangle &=\frac{1}{|G|}\sum_{x\in G}\Ind_H^G \chi(g) \overline{\psi(g)}\\
  &=\frac{1}{|G||H|} \sum_{g\in G} \big(\sum_{x\in G} \tilde{\chi}(x^{-1}gx)\big)\overline{\psi(g)}\\
  &=\frac{1}{|G||H|} \sum_{x\in G} \big(\sum_{g\in G}\tilde{\chi}(x^{-1}gx)\overline{\psi(g)}\big)\\
  &=\frac{1}{|G||H|} \sum_{x\in G} \big(\sum_{g'\in G}\tilde{\chi}(g')\overline{\psi(xg'x^{-1})}\big)\\
  &=\frac{1}{|G||H|} \sum_{x\in G} \big(\sum_{g'\in  G}\tilde{\chi}(g')\overline{\psi(g')}\big)\\
  &=\frac{1}{|H|} \sum_{h\in H}\chi(h)\overline{\psi(h)}= \langle \chi, \Res_H^G \psi\rangle.
\end{align*}
\end{proof}

\subsection{Examples from coding theory} 
In coding theory we are mostly interested in the decomposition of $\CC(X)$ under the
action of $G=\Aut(X)$ for various spaces $X$. We recall that the
action of $G$ on $f\in \CC(X)$ is given by $(gf)(x)=f(g^{-1}x)$.
The space $\CC(X)$ is endowed with the inner product
\begin{equation*}
\langle f, f'\rangle =\frac{1}{|X|}\sum_{x\in X} f(x)\overline{f'(x)}.
\end{equation*}
which is $G$-invariant.

\subsubsection{\bf The binary Hamming space $H_n$:}\label{Hamming 1} recall that $G=T\rtimes S_n$.
Let, for $y\in H_n$, $\chi_y\in \CC(H_n)$ be defined by $\chi_y(x)=(-1)^{x\cdot y}$. 
The set  $\{\chi_y,y\in H_n\}$ is exactly the set of
irreducible characters of the additive group $\F_2^n$, and form an
orthonormal
basis of $\CC(H_n)$.
The computation of the action of $G$ on $\chi_y$ shows that
for $\sigma\in S_n$, $\sigma\chi_y=\chi_{\sigma(y)}$
and for $t_u\in T$, $t_u\chi_y=(-1)^{u\cdot y} \chi_y$.
Let, for $0\leq k\leq n$,
\begin{equation*}
P_k:=\perp_{y , wt(y)=k} \C \chi_y
\end{equation*}
Thus $P_k$ is a $G$-invariant subspace of $\CC(H_n)$ of dimension
$\binom{n}{k}$ and we have the decomposition
\begin{equation}\label{dec Hamming}
\CC(H_n)=P_0\perp P_1\perp\dots \perp P_n.
\end{equation}
The computation  $\langle \chi_{P_k},\chi_{P_k}\rangle=1$ where
$\chi_{P_k}$
is the character of the $G$-module $P_k$ shows that
these modules are $G$-irreducible.

Now we introduce the Krawtchouk polynomials. The element 
$Z_k:=\sum_{wt(y)=k} \chi_y$ of $\CC(H_n)$ is $S_n$-invariant. In other
words,
$Z_k(x)$ only depends on $wt(x)$. We define the Krawtchouk polynomial
$K_k$ for $0\leq k\leq n$
by 
\begin{align}\label{Krawtchouk}
K_k(w):&=Z_k(x)=\sum_{wt(y)=k} (-1)^{x\cdot y}\text { where }
wt(x)=w\\*
&=\sum_{i=0}^w (-1)^i\binom{w}{i}\binom{n-w}{k-i}.
\end{align}

We review some properties of these polynomials:
\begin{enumerate}
\item $\deg(K_k)=k$
\item $K_k(0)=\binom{n}{k}$
\item Orthogonality relations: for all $0\leq k\leq l\leq n$
\begin{equation*}
\frac{1}{2^n}\sum_{w=0}^n \binom{n}{w}K_k(w)K_l(w)=\delta_{k,l}\binom{n}{k}
\end{equation*}
\end{enumerate}
The last property is just a reformulation of the orthogonality of the
$Z_k\in P_k$, since, if $f,f'\in \CC(H_n)$ are $S_n$-invariant, and
$\tilde{f}(w):=f(x)$, $wt(x)=w$,
\begin{align*}
\langle f,f'\rangle&=\frac{1}{2^n}\sum_{x\in H_n} f(x)f'(x)\\
&=\frac{1}{2^n}\sum_{w=0}^n \binom{n}{w}\tilde{f}(w)\tilde{f'}(w).
\end{align*}
The above three properties characterize uniquely the Krawtchouk
polynomials. 

Let $C\subset H_n$ be a binary code. Let $\1_C$ be the characteristic
function of $C$. The obvious inequalities hold:
\begin{equation}\label{pos1}
0\leq k\leq n,\quad \sum_{wt(y)=k} \langle \1_C,\chi_y\rangle^2\geq 0.
\end{equation}
Since the decomposition of $\1_C$ over the basis $\chi_y$
reads
\begin{equation*}
\1_C=\sum_{y\in H_n} \langle \1_C,\chi_y\rangle \chi_y.
\end{equation*}
the above inequalities are indeed  reformulations  of  the non negativity of the
squared norm of the projections $p_{P_k}(\1_C)$. They express in terms
of the Krawtchouk polynomials:
\begin{equation}\label{pos2}
0\leq k\leq n,\quad \frac{1}{2^{2n}}\sum_{(x,x')\in C^2} K_k(d_H(x,x'))\geq 0
\end{equation}
or equivalently in terms of the distance distribution of the code $C$:
if
\begin{equation*}
A_w(C):=\frac{1}{|C|}|\{(x,x')\in C^2 : d_H(x,x')=w\}|
\end{equation*}
then 
\begin{equation*}\label{pos3}
0\leq k\leq n,\quad \frac{|C|}{2^{2n}}\sum_{w=0}^n A_w(C)K_k(w)\geq 0.
\end{equation*}
These inequalities are the basic inequalities involved in Delsarte
linear programming method. We shall encounter similar inequalities in
a very general setting.

In the special case when $C$ is linear, we have 
\begin{equation*}
\langle \1_C,\chi_y\rangle=\frac{|C|}{2^n}\1_{C^{\perp}}(y)
\end{equation*}
so that we recognise the identity 
\begin{equation*}
\sum_{wt(y)=k} \langle \1_C,\chi_y\rangle^2=\frac{|C|}{2^{2n}}\sum_{w=0}^n A_w(C)K_k(w)
\end{equation*}
to be the Mac Williams identity
\begin{equation*}
A_k(C^{\perp})=\frac{1}{|C|}\sum_{w=0}^n A_w(C)K_k(w).
\end{equation*}

The more general $q$-ary Hamming space affords similar results; it is
treated in \ref{qHamming 1}.

\subsubsection{\bf The Johnson spaces $J_n^w$:}\label{Johnson 1}  the
group is $G=S_n$.
Here, we shall see at work a standard way to evidence $G$-submodules as kernels of
$G$-endomorphisms. For details we refer to \cite{Del2} where the
$q$-Johnson spaces are given a uniform treatment. We introduce the
applications
\begin{align*}
\delta : \CC(J_n^w)&\to \CC(J_n^{w-1})\\
            f&\mapsto \delta(f) : \delta(f)(x):=\sum_{y\in J_n^w,\ 
              x\subset y} f(y)
\end{align*}
and 
\begin{align*}
\psi : \CC(J_n^{w-1})&\to \CC(J_n^{w})\\
            f&\mapsto \psi(f) : \psi(f)(x):=\sum_{y\in J_n^{w-1},\ 
              y\subset x} f(y)
\end{align*}
Both of these applications commute with the action of $G$. They
satisfy the following properties: 
$\langle f, \psi(f')\rangle = \langle \delta(f), f'\rangle$, $\psi$ is
injective  and $\delta$ is surjective. 
Therefore the subspace of $\CC(J_n^w)$:
\begin{equation*}
H_w:=\ker \delta
\end{equation*}
is a $G$-submodule of dimension $\binom{n}{w}-\binom{n}{w-1}$ 
and we have the orthogonal decomposition
\begin{equation*}
\CC(J_n^w)=H_w\perp \psi(\CC(J_n^{w-1}))\simeq H_w\perp \CC(J_n^{w-1}).
\end{equation*}
By induction we obtain a decomposition
\begin{equation*}
\CC(J_n^w)\simeq H_w\perp H_{w-1}\perp\dots \perp H_0
\end{equation*}
which can be proved to be the irreducible decomposition of
$\CC(J_n^w)$ (see \ref{Johnson 2}).

\section{Linear representations of compact groups}

In this section we enlarge the discussion to the representation theory
of compact groups. For this section we refer to \cite{Bump}.

\subsection{Finite dimensional representations}

The theory of finite dimensional representations of finite groups extends nicely and
straightforwardly to compact groups. 
A finite dimensional representation of a compact group $G$ is a
continuous homomorphism $\rho: G\to \Gl(V)$ where $V$ is a complex
vector space of finite dimension. 

A compact group $G$
affords a Haar measure, which is a regular left and right
invariant measure. We assume this measure to be normalized, i.e. the
group $G$ has measure $1$. With this measure the finite sums
over elements of a finite group can be replaced with integrals; so 
the crucial construction of a $G$-invariant inner product in the
proof of Maschke theorem
extends to compact groups with the formula
\begin{equation*}
\langle x,y\rangle':=\int_G \langle gx,gy\rangle dg.
\end{equation*}
Hence Maschke theorem remains valid for  finite dimensional
representations.
We keep the notation $\RR$ for a set of representatives of the
finite dimensional irreducible representations of $G$, chosen to be
representations with unitary matrices. 
A main
difference with the finite case is that $\RR$ is not finite anymore.

\subsection{Peter Weyl theorem}

Infinite dimensional representations  immediately occur 
with the generalization of permutation representations. 
Indeed, if $G$ acts continuously on a space $X$, it is natural to
consider the action of $G$ on the space $\CC(X)$ of complex valued
continuous functions on $X$ given by $(gf)(x)=f(g^{-1}x)$ to be  
a natural generalization of permutation representations. 
A typical example of great interest in
coding theory is the action of $G=O(\R^n)$ on the 
unit sphere of the Euclidean space:
\begin{equation*}
S^{n-1}:=\{x\in \R^n \ :\ x\cdot x=1\}.
\end{equation*}
The regular representation, which is the special case $\CC(G)$, with
the left action of $G$ on itself, can be
expected to play an important role similar to the finite case.
It is endowed with the inner product 
\begin{equation*}
\langle f,f'\rangle:=\int_{G} f(g)\overline{f'(g)}dg.
\end{equation*}
For  $R\in \RR$, the matrix coefficients  $g\to
R_{i,j}(g)$ belong to unitary matrices.
Theorem \ref{orthog rel} establishing the orthogonality relations
between the matrix coefficients of the elements of $\RR$ remains
valid;
thus they form an orthogonal system in $\CC(G)$. The celebrated Peter Weyl theorem asserts that these elements span a vector space
which is dense in $\CC(G)$ for the topology of uniform convergence.

\begin{theorem}\label{PWT}[Peter Weyl theorem]
The finite  linear combinations of the functions $R_{i,j}$ are dense
in $\CC(G)$ for the topology of uniform convergence.
\end{theorem}

\begin{proof} We give a sketch of the proof:
\begin{enumerate}
\item If $V$ is a finite dimensional  subspace of $\CC(V)$
which is stable by right translation (i.e. by $gf(x)=f(xg)$) 
and $f\in V$, then $f$ is a linear
combination of a finite number of the $R_{i,j}$: according to previous 
discussion, there is a decomposition $V=W_1\oplus\dots\oplus W_n$ such
that
$W_k$ is irreducible. If $W_k\simeq R$, there exists a basis 
$e_1,\dots, e_{d_R}$ of $W_k$ in which the action of $G$ has
matrices $R$.  Explicitly,
\begin{equation*}
e_j(hg)=\sum_{i=1}^{d_R} R_{i,j}(g) e_i(h).
\end{equation*}
Taking $h=1$, we obtain $e_j=\sum_{i=1}^{d_R} e_i(1) R_{i,j}$.
\item 
The idea is to approximate $f\in \CC(G)$ by elements of
such subspaces,  constructed from the eigenspaces of a 
compact selfadjoint operator. We introduce the convolution operators: 
let $\phi\in\CC(G)$,
\begin{equation*}
T_{\phi} (f)(g)=(\phi * f)(g)=\int_{G} \phi(gh^{-1})f(h) dh.
\end{equation*}
\item Since $G$ is compact, $f$ is  uniformly  continuous; this
  property allows to choose $\phi$ 
  such that $\| f-T_{\phi}(f)\|_{\infty}$ is arbitrary small. 
\item The operator $T_{\phi}$ is compact and can be assumed to be
  selfadjoint. The spectral theorem for such operators 
on Hilbert spaces (here $L^2(G)$) asserts that the eigen\-spaces $V_{\lambda}:=\{f:
T_{\phi}f=\lambda f\}$ for $\lambda\neq 0$ are finite dimensional 
and that the space is the direct Hilbert sum $\oplus_{\lambda} V_{\lambda}$.
For $t>0$, the subspaces
$V_t:=\oplus V_{\lambda,\ |\lambda| > t}$
have finite dimension (i.e. there is only a finite
number of eigenvalues $\lambda$ with $|\lambda|>t>0$). 
\item The operator $T_{\phi}$ commutes with the action of $G$ by
  right translation thus the subspaces $V_{\lambda}$ are stable under  this action.
\item Let $f_{\lambda}$ be the projection of $f$ on $V_{\lambda}$. The
  finite sums $f_t:=\sum_{|\lambda|> t} f_{\lambda}$ 
converge to $f-f_0$
  for the $L^2$-norm when $t\to 0$.
\item Moreover, for all $f\in \CC(V)$, $\| T_{\phi}(f)\|_{\infty} \leq
  \|\phi\|_{\infty}\|f\|_2$.
Thus, $T_{\phi}(f_t)$ converges {\it uniformly} to
$T_{\phi}(f-f_0)=T_{\phi}(f)$. Finally, $T_{\phi}(f_t)\in V_t$ and $V_t$
  is finite dimensional and invariant under the action of $G$ by right translations, thus by
  (1) $T_{\phi}(f_t)$ is a linear
  combinations of the $R_{i,j}$.
\end{enumerate}

\end{proof}

If $d_R=\dim(R)$, the vector space spanned by $\{\overline{R_{i,j}},
i=1,\dots,d_R\}$ is $G$-invariant and isomorphic to
$R$. So Peter-Weyl theorem means that the decomposition of the regular
decomposition is
\begin{equation*}
\CC(G)=\perp_{R\in \RR} \I_R 
\end{equation*}
where $\I_R\simeq R^{d_R}$,
generalizing Theorem \ref{Dec Reg Rep} (one has a better understanding
of this decomposition with the action of $G\times G$ on $G$
given by $(g,g')h=ghg'^{-1}$. For this action $\CC(G)=\oplus_{R\in
  \RR} R\otimes R^*$ where $R^*$ is the contragredient representation,
and $R\otimes R^*$ is $G\times G$-irreducible).

Since uniform convergence is stronger than $L^2$ convergence, we also
have as a consequence of Peter Weyl theorem that the matrix
coefficients $R_{i,j}$ (suitably rescaled) form an orthonormal basis
of $L^2(G)$  in the sense
of Hilbert spaces.

A slightly more general version of
Peter Weyl theorem deals with the decomposition of $\CC(X)$ where $X$
is a compact space on which $G$ acts homogeneously. If $G_{x_0}$ is
the stabilizer of a  base point $x_0\in X$, then $X$ can be
identified with the quotient space $G/G_{x_0}$. The Haar measure
on $G$ gives rise to a
$G$-invariant regular measure $\mu$ on $X$ and   $\CC(X)$ 
is endowed with the inner product 
\begin{equation*}
\langle f,f'\rangle:=\frac{1}{\mu(X)}\int_X f(x)\overline{f'(x)}d\mu(x).
\end{equation*}
The space $\CC(X)$ can be identified with the space
$\CC(G)^{G_{x_0}}$
of $G_{x_0}$-invariant (for the right translation) functions thus $\CC(X)$ affords a
decomposition of the form
\begin{equation*}
\CC(X)\simeq \perp_{R\in \RR} R^{m_R}
\end{equation*}
for some integers $m_R$, $0\leq m_R\leq d_R$, in the sense of uniform as well as  $L^2$
convergence.

A more serious generalization of the above theorem deals with the
unitary representations of $G$. These are the continuous
homomorphisms from $G$ to the unitary group of a Hilbert space.

\begin{theorem}\label{Hilbert}
Let $\pi:G\to U(H)$ be a continuous homomorphism from $G$ to the
unitary group of a Hilbert space $H$. Then $H$ is a direct Hilbert sum
of finite dimensional irreducible $G$-modules.
\end{theorem}

\begin{proof}
The idea is to construct in H a $G$-subspace of
finite dimension and then to iterate with the orthogonal complement of
this subspace.
To that end, for a fixed $v\in H$, one chooses $f\in \CC(G)$ such that
$\int_{G} f(g) (\pi(g)v) dg\neq 0$. From Peter Weyl theorem, $f$ can be
assumed to be a finite linear combination of the $R_{i,j}$. In other
words,
there exists a finite dimensional unitary  representation $(\rho,V)$
and $e_1,e_2\in V$ such
that $f(g)=\langle \rho(g^{-1})e_1,e_2\rangle_V$.
The operator $T:V\to H$ defined by
\begin{equation*}
T(x)=\int_G \langle \rho(g^{-1})x,e_2\rangle_V  (\pi(g)v )dg
\end{equation*}
commutes with the actions of $G$ and is non zero. Thus its image is
a non zero $G$-subspace of finite dimension of $H$.

\end{proof}

\subsection{Examples}

\subsubsection{\bf The unit sphere $S^{n-1}$:}\label{sphere 1} it is the basic
example. The orthogonal group 
$G=O(\R^n)$ acts homogeneously  on $S^{n-1}$. The stabilizer $G_{x_0}$ of $x_0$ can be
identified with $O(x_0^{\perp})\simeq O(\R^{n-1})$.
Here $\mu=\omega $ is the Lebesgue measure on $S^{n-1}$.
We set $\omega_n:=\omega(S^{n-1})$. The irreducible decomposition of
$\CC(S^{n-1})$ is as follows:
\begin{equation*}\label{dec 1}
\CC(S^{n-1})=H^n_0\perp H^n_1\perp\dots H^n_k\perp \dots 
\end{equation*}
where $H^n_k$ is isomorphic to the space $\Harm_k^n$ of harmonic
polynomials:
\begin{equation*}
\Harm^n_k:=\{P\in \C[X_1,\dots, X_n]_k : \Delta P=0, \Delta=\sum_{i=1}^n
\frac{\partial^2}{\partial x_i^2}\}
\end{equation*}
The space $\Harm_k^n$ is a $O(\R^n)$-module because the Laplace operator
$\Delta$
commutes with the action of the orthogonal group and it is moreover irreducible.
Its dimension equals $h_k^n:=\binom{n+k-1}{k}-\binom{n+k-3}{k-2}$. The embedding of
$\Harm_k^n$ into $\CC(S^{n-1})$ is the obvious one, to the
corresponding polynomial
function in the $n$ 
coordinates.

\subsubsection{\bf The action of stabilizers of many points:}\label{sphere 3} 
for our purposes we are interested in the decomposition of some
  spaces $\CC(X)$, $X$ homogeneous for $G$,  for the action of a
  subgroup $H$ of $G$, typically $H=G_{x_1,\dots,x_s}$ the
  stabilizer of $s$ points. In order to describe it, it is enough to 
study the decomposition of the $G$-irreducible submodules of
$\CC(X)$ under the action of $H$; thus we have to decompose only
finite dimensional spaces. However, because the same irreducible
representation of $H$ may occur in infinitely many of the $G$-isotypic
subspaces, it happens that the $H$-isotypic subspaces are not of
finite dimension.  A typical example is given by $X=S^{n-1}$,
$G=O(\R^n)$ and $H=G_e\simeq O(\R^{n-1})$. 
It is a classical result that for the
restricted action to $H$ the decomposition of $\Harm_k^n$ into
$H$-irreducible subspaces is given by:
\begin{equation}\label{dec res}
\Harm_k^n\simeq \bigoplus_{i=0}^k \Harm_i^{n-1}.
\end{equation}
Hence, each of the $H_k^n$ in \eqref{dec 1} decomposes likewise:
\begin{equation*}
H^n_k=H_{0,k}^n\perp H^n_{1,k}\perp\ldots \perp H^n_{k,k}
\end{equation*}
where $H^n_{i,k}\simeq \Harm_i^{n-1}$. We have  the following picture,
where the $H$-isotypic components appear to be the rows of the second decomposition.

\[
\begin{array}{cccccccccc}
\CC(S^{n-1}) & =_G & H^n_0         & \perp & H^n_1         & \perp & \ldots &
\perp & H^n_k&\perp \ldots \\
         & =_H & H^n_{0,0} & \perp & H^n_{0,1} & \perp & \ldots & \perp &  H^n_{0,k}&\perp \ldots\\
         &   &               & \perp & H^n_{1,1} & \perp & \ldots & \perp & H^n_{1,k}&\perp \ldots\\
&   &               &       &               & \multicolumn{5}{r}{\cdots\cdots\cdots\cdots\cdots\cdots\cdots\cdots}\\
         &   &               &       &               &        &  &
         \perp & H^n_{k,k}&\perp \ldots\\
\end{array}
\]

\section{Harmonic analysis of compact spaces}\label{Harmonic analysis}

We take notations for the rest of the lecture notes. 
$X$ is a compact space (possibly finite) on which a compact group
(possibly finite) $G$ acts continuously. Moreover,
$X$ is endowed with
a $G$-invariant Borel regular measure $\mu$ for which $\mu(X)$ is finite. If $X$ itself is finite, the topology is the discrete topology and
the measure is the counting measure. In the previous sections we 
have discussed the decomposition of the permutation representation
$\CC(X)$. In order to lighten the notations, we assume that $G$ has a
countable number of finite dimensional irreducible representations
(it is the case if $G$ is a group of matrices over the reals since
then $L^2(G)$ is a separable Hilbert space),
and we let $\RR=\{R_k,k\geq 0\}$, where $R_0$ is the trivial
representation. We let $d_k:=\dim(R_k)$.
From Theorem \ref{Hilbert}, we have a decomposition
\begin{equation}\label{DEC}
\CC(X)\subset L^2(X)=\oplus_{k\geq 0, 1\leq i\leq m_k} H_{k,i}
\end{equation}
where $H_{k,i}\subset \CC(X)$, $H_{k,i}\simeq R_k$, $0\leq m_k\leq +\infty$ (the case $m_k=0$
means that $R_k$ does not occur, the case $m_k=+\infty$ may occur if
$G$ is not transitive on $X$). The isotypic subspaces are pairwise
orthogonal and denoted
$\I_k$:
\begin{equation*}
\I_k=\oplus_{i=1}^{m_k} H_{k,i}
\end{equation*}
We take the subspaces $H_{k,i}$ to be  also pairwise orthogonal.
For all $k,i$, we choose an orthonormal basis
$e_{k,i,1},\dots,e_{k,i,d_k}$ of $H_{k,i}$ such that in this basis the
  action of $g\in G$ is expressed by the unitary matrix  $R_k(g)$.
The set $\{e_{k,i,s}\}$ is an orthonormal basis in the Hilbert sense.

\subsection{Commuting endomorphisms and zonal matrices.}
In this subsection we want to give more information on the algebra
$\End_G(\CC(X))$ of
 commuting continuous endomorphisms of $\CC(X)$. 
We introduce, for $K\in \CC(X^2)$, the operators $T_K$, called
Hilbert-Schmidt operators:
\begin{equation*}
T_K(f)(x)=\frac{1}{\mu(X)} \int_X K(x,y)f(y)d\mu(y).
\end{equation*}
It is easy to verify that $T_K\in \End_G(\CC(X))$ if $K$ is $G$-invariant,
i.e. if $K(gx,gy)=K(x,y)$ for all $g\in G$, $(x,y)\in X^2$. A
continuous function $K(x,y)$ with this property 
is also called a zonal function. It is also easy,
but worth to notice that $T_K\circ T_{K'}=T_{K*K'}$ where $K*K'$ is
the convolution of $K$ and $K'$:
\begin{equation*}
(K*K')(x,y):=\int_X K(x,z)K'(z,y)d\mu(z).
\end{equation*}
Let
\begin{equation*}
\K:=\{K\in \CC(X^2) : K(gx,gy)=K(x,y)\text{ for all }g\in G,
(x,y)\in X^2\}.
\end{equation*}
The triple $(\K,+,*)$ is a $\C$-algebra (indeed a $\C^*$-algebra, with
$K^*(x,y):=\overline{K(y,x)}$). Thus we have an embedding
$\K\to \End_G(\CC(X))$.

Assume $V\subset \CC(X)$ is a finite dimensional
 $G$-subspace such that $V=W_1\perp \dots \perp W_m$ with $W_i\simeq R\in
 \RR$. By the same proof as the one of  Theorem \ref{End},
 $\End_G(V)\simeq \C^{m\times m}$. More precisely, we have seen that,
 if $u_{j,i}:W_i\to W_j$ are $G$-isomorphisms, such that $u_{k,j}\circ
 u_{j,i}=u_{k,i}$, then an element
 $\phi\in \End_G(V)$ is associated to a matrix $A=(a_{i,j})\in
 \C^{m\times m}$ such that, for all $f\in V$, with $p_{W_i}(f)=f_i$,
\begin{equation*}
\phi(f)=\sum_{i,j=1}^m a_{j,i}u_{j,i}(f_i).
\end{equation*}
For all $1\leq i\leq m$, let $(e_{i,1},\dots ,e_{i,d})$, $d=\dim(R)$,  be an
 orthonormal basis of $W_i$ such that in this basis the
  action of $g\in G$ is expressed by the unitary matrix  $R(g)$.
We define
\begin{equation*}
E_{i,j}(x,y):=\sum_{s=1}^{d}e_{i,s}(x)\overline{e_{j,s}(y)}.
\end{equation*}
Then we have:
\begin{lemma}\label{E1} The above defined functions $E_{i,j}$ satisfy:
\begin{enumerate}
\item $E_{i,j}$ is zonal: $E_{i,j}(gx,gy)=E_{i,j}(x,y)$. 
\item Let $T_{i,j}:=T_{E_{i,j}}$. Then $T_{j,i}(W_i)=W_j$ and $T_{j,i}(W_k)=0$ for $k\neq i$.
\item $T_{i,j}\circ T_{j,k}=T_{i,k}$.
\end{enumerate}
\end{lemma}

\begin{proof} 

\begin{enumerate}
\item From the construction, we have
\begin{equation*}
e_{i,s}(gx)=\sum_{t=1}^d \overline{R_{s,t}(g)}e_{i,t}(x)
\end{equation*}
thus
\begin{align*}
E_{i,j}(gx,gy) &= \sum_{s=1}^de_{i,s}(gx)\overline{e_{j,s}(gy)}\\
&=\sum_{s=1}^d \sum_{k,l=1}^d \overline{R_{s,k}(g)}R_{s,l}(g)e_{i,k}(x)\overline{e_{j,l}(y)}\\
&=\sum_{k,l=1}^d\Big(\sum_{s=1}^d
\overline{R_{s,k}(g)}R_{s,l}(g)\Big)e_{i,k}(x)\overline{e_{j,l}(y)}\\
&=\sum_{k}^de_{i,k}(x)\overline{e_{j,k}(y)}=E_{i,j}(x,y)\\
\end{align*}
where the second last equality holds because $R(g)$ is a
unitary matrix.
\item We compute $T_{j,i}(e_{k,t})$:
\begin{align*}
T_{j,i}(e_{k,t})(x) &= \frac{1}{\mu(X)} \int_X\big(\sum_{s=1}^de_{j,s}(x)\overline{e_{i,s}(y)}\big)e_{k,t}(y)d\mu(y)\\
&=\frac{1}{\mu(X)}\sum_{s=1}^de_{j,s}(x)\int_X\overline{e_{i,s}(y)}e_{k,t}(y)d\mu(y)\\
&= \sum_{s=1}^de_{j,s}(x)\langle e_{k,t}, e_{i,s} \rangle\\
&= \sum_{s=1}^de_{j,s}(x)\delta_{k,i}\delta_{t,s}=\delta_{k,i}e_{j,t}(x).
\end{align*}
\item Similarly one computes that
\begin{equation*}
E_{i,j}*E_{l,k}=\delta_{j,l} E_{i,k}.
\end{equation*}
\end{enumerate}
\end{proof}

The $E_{i,j}(x,y)$ put together form a
matrix $E=E(x,y)$, that we call the zonal matrix associated to the
$G$-subspace $V$:
\begin{equation}\label{EE}
E(x,y):=\big( E_{i,j}(x,y) \big)_{1\leq i,j\leq m}.
\end{equation}
At this stage is is natural to discuss the dependence of this matrix 
on the various ingredients needed for its definition.
\begin{lemma}\label{E2} We have
\begin{enumerate}
\item $E(x,y)$ is unchanged if another orthonormal basis of $W_i$
  is chosen (i.e. if another unitary representative of the irreducible
  representation $R$ is chosen).
\item $E(x,y)$ is changed to $AE(x,y)A^*$ for some matrix $A\in \Gl(\C^m)$
  if another decomposition (not necessarily with orthogonal spaces) $V=W'_1\oplus\dots \oplus W'_m$ is chosen.
\end{enumerate}
\end{lemma}

\begin{proof}
\begin{enumerate}
\item Let $(e'_{i,1},\dots, e'_{i,d})$ be another orthonormal basis of
  $W_i$ and let  $U_i$ be unitary $d\times d$ matrices  such that 
\begin{equation*}
(e'_{i,1},\dots, e'_{i,d})=(e_{i,1},\dots, e_{i,d})U_i.
\end{equation*}
Since we want the representation $R$ to be realized by the same
matrices in the basis $(e'_{i,1},\dots, e'_{i,d})$ when $i$ varies, we
have $U_i=U_j=U$. Then, with obvious notations,
\begin{align*}
E'_{i,j}(x,y)=&(e'_{i,1}(x),\dots, e'_{i,d}(x))(e'_{i,1}(y),\dots, e'_{i,d}(y))^*\\
=&(e_{i,1}(x),\dots, e_{i,d}(x))UU^*(e_{i,1}(y),\dots, e_{i,d}(y))^*\\
=&(e_{i,1}(x),\dots, e_{i,d}(x))(e_{i,1}(y),\dots, e_{i,d}(y))^*\\
=&E_{i,j}(x,y).
\end{align*}

\item If $V=W_1\perp \dots \perp W_m=W'_1\perp \dots \perp W'_m$ with
  basis
$(e_{i,1},\dots,e_{i,d})$ of $W_i$ and $(e'_{i,1},\dots,e'_{i,d})$ of
  $W'_i$ 
in which the action of $G$ is by the same matrices $R(g)$, let
$\phi\in \End(V)$ be defined by $\phi(e_{i,s})=e'_{i,s}$. Clearly
$\phi$ commutes with the action of $G$; if $u_{j,i}$ is defined by
$u_{j,i}(e_{i,s})=e_{j,s}$ then we have seen that, for some matrix
$A=(a_{i,j})$,
$e'_{i,s}=\phi(e_{i,s})=\sum_{j=1}^m a_{j,i} e_{j,s}$. Moreover $A$ is
invertible. It is unitary if the spaces $W'_i$ are pairwise orthogonal. With the notations $E(x):=(e_{i,s}(x))$, we have
\begin{equation*}
E(x,y)=E(x)E(y)^*\text{ and } E'(x)=A^t E(x)
\end{equation*}
thus
\begin{equation*}
E'(x,y)=A^t E(x,y)\overline{A}.
\end{equation*}

\end{enumerate}
\end{proof}

Going back to $\phi\in \End_G(V)$, from Lemma \ref{E1} we can take
$u_{j,i}=T_{j,i}$ and we have the expression
\begin{equation*}
\phi=\sum_{i,j=1} ^m a_{j,i}T_{j,i} =T_{\langle A, \overline{E}\rangle}.
\end{equation*}
We take the following notation: the space of linear combinations of
elements of the form $f(x)\overline{g(y)}$ for $(f,g)\in V^2$ is
denoted $V^{(2)}$. We have proved the following:
\begin{proposition}\label{End2}
Let 
\begin{equation*}
\K_V:=\{K\in V^{(2)} : K(gx,gy)=K(x,y) \text{ for all }g\in G,
(x,y)\in X^2 \}.
\end{equation*}
The following are
isomorphisms of $\C$-algebras:
\begin{equation*}
\begin{array}{rlcrl}
\K_V & \to \End_G(V)&\qquad &\C^{m\times m} & \to \End_G(V)\\
K &\mapsto T_K &\qquad & A& \mapsto T_{\langle A, \overline{E}\rangle}.
\end{array}
\end{equation*}
Moreover, $\End_G(\CC(X))$ is commutative iff $\K$ is commutative iff
$m_k=0,1$ for all $k\geq 0$.
\end{proposition}

\begin{proof} The isomorphisms are clear from previous discussion. For
  the last assertion, it is enough to point out that
\begin{equation*}
\End_G(\CC(X))=\prod_{k\geq 0 } \End_G(\I_k).
\end{equation*}
\end{proof}

\begin{remark}\label{Rem} Proposition \ref{End2} shows in particular that $\K_V$
  and $\End_G(V)$ have the same dimension. It is sometimes easy to
  calculate the dimension of $\K_V$; for example if $X$ is a finite
  set and $V=\CC(X)$, then $\dim(\K_V)$ is exactly equal to the number of orbits of $G$ acting on
  $X^2$. On the other hand, in this case, the dimension of $\End_G(V)$
  is
the sum of the squares of the multiplicities in $\CC(X)$ of the irreducible
representations of $G$. For the binary Hamming space
treated in \ref{Hamming 1}, the orbits of $G$ acting on $X^2$ are in one to one
correspondance with the values taken by the Hamming distance, i.e. there
are $(n+1)$ such orbits. Thus, once we have obtained the decomposition
$\CC(H_n)=P_0\perp\dots \perp P_n$, because this decomposition involves allready
$(n+1)$ subspaces, we can conclude readily that these subspaces are irreducible.
This reasoning applies also to the Johnson space \ref{Johnson 1} and to the more
general $q$-Hamming space \ref{qHamming 1}.
A variant of this method is as follows: if we suspect $V\subset
\CC(X)$ to be irreducible, then it is enough to prove that $\K_V$ has
dimension $1$. See in \ref{Johnson 2} for an illustration.
\end{remark}

\subsection{Examples: $G$-symmetric spaces.}

\begin{definition}
We say that $X$ is
$G$-symmetric if for all $(x,y)\in X^2$, there exists $g\in G$ such
that $gx=y$ and $gy=x$. In other words, $(x,y)$ and $(y,x)$ belong to
the same orbit of $G$ acting on $X^2$.
\end{definition}
A first consequence of Proposition \ref{End2} is that $G$-symmetric
spaces have multiplicity free decompositions. 

\begin{proposition}\label{commute}
If $X$ is $G$-symmetric then $m_k=0,1$ for all $k\geq 0$ and
$E_k(x,y)$ is real symmetric.
\end{proposition}

\begin{proof} For all $K\in \K$, $K(x,y)=K(y,x)$. Thus $\K$ is
  commutative:
indeed,
\begin{align*}
(K'* K)(x,y) &=\frac{1}{\mu(X)}\int_X K'(x,z)K(z,y) d\mu(z) \\
&=\frac{1}{\mu(X)}\int_X K'(z,x)K(y,z) d\mu(z) \\
&=(K*K')(y,x)=(K*K')(x,y).
\end{align*}
Moreover $\overline{E_k(x,y)}=E_k(x,y)=E_k(y,x)$.
\end{proof}

\subsubsection{\bf $2$-point homogeneous spaces:}\label{2-homogeneous 1}
these spaces are prominent  examples of
$G$-symmetric spaces.
\begin{definition}
A metric spaces $(X,d)$  is said to be $2$-point homogeneous  for the
action of $G$ if $G$ is transitive on $X$, leaves the distance $d$
invariant, and if, for $(x,y)\in X^2$,
\begin{equation*}
\text{there exists }g\in G \text{ such that }(gx,gy)=(x',y')
\Longleftrightarrow d(x,y)=d(x',y').
\end{equation*}
\end{definition}
Examples of such spaces of interest in coding theory are numerous:
the Hamming and Johnson spaces,  endowed with the Hamming distance,
for the action of respectively $T\rtimes S_n$ and $S_n$;
the unit sphere $S^{n-1}$ for the angular distance $\theta(x,y)$ and
the action of the orthogonal group.
It is a classical result that, apart from $S^{n-1}$, the projective
spaces $\Pp^n(K)$ for $K=\R, \C, \mathbb{H}$, and $\Pp^2(\OO)$, are
  the only real compact $2$-point homogeneous spaces.

There are more examples of finite $2$-point homogeneous spaces, we can
mention among them the $q$-Johnson spaces. The $q$-Johnson space
$J_n^w(q)$ is the set of linear
subspaces of $\F_q^n$ of fixed dimension $w$, with the action of the
group
$\Gl(\F_q^n)$ and the distance $d(x,y)=\dim(x+y)-\dim(x\cap y)$.
We come back to this space in the next section.

\smallskip
There are other symmetric spaces occurring in coding theory: 

\subsubsection{\bf The Grassmann spaces:}\label{Grassmann 1} $X=\Gmn(K)$, $K=\R,\C$, i.e. the set of
  $m$-dimensional linear subspaces of $K^n$, with the
  homogeneous action of
  $G=O(\R^n)$ (respectively $U(\C^n)$). This space is
  $G$-symmetric but not $2$-point homogeneous (if $m\geq 2$). The orbits of $G$ acting on
  pairs
$(p,q)\in X^2$ are characterized by their principal angles \cite{G}. The
  principal angles of $(p,q)$ are $m$ angles $(\theta_1,\dots,
  \theta_m)\in [0,\pi/2]^m$ constructed as follows: one iteratively
  constructs an orthonormal basis $(e_1,\dots, e_m)$ of $p$ and 
an orthonormal basis $(f_1,\dots, f_m)$ of $q$ such that, for $1\leq
i\leq m$, 
\begin{equation*}
\begin{array}{lll}
\cos\theta_i&=\max \{|(e,f)|\ :\  &e\in p,\ f\in q,\\
&&(e,e)=(f,f)=1, \\
&& (e,e_j)=(f,f_j)=0 \text{ for }1\leq j\leq i-1\}\\
&=|(e_i,f_i)|&
\end{array}
\end{equation*}
The we have (see \cite{G}):
\begin{equation*}
\begin{array}{c}
\text{there exists }g\in G \text{ such that }(gp,gq)=(p',q')\\
\Longleftrightarrow\\
(\theta_1(p,q),\dots,\theta_m(p,q))=(\theta_1(p',q'),\dots,\theta_m(p',q')).
\end{array}
\end{equation*}

 \subsubsection{\bf The ordered Hamming space:}\label{ordered Hamming
   1} $X=(\F_2^r)^n$ (for the sake of simplicity we restrict here to
 the binary case). 
Let $x=(x_1,\dots,
  x_n)\in X$ with $x_i\in \F_2^r$. For $y\in \F_2^r$, the ordered weight of
  $y$, denoted $w_r(y)$, is the right most non zero coordinate of
  $y$. The ordered weight of $x\in X$ is $w_r(x):=\sum_{i=1}^n
  w_r(x_i)$ and the ordered distance of two elements $(x,y)\in X^2$ is
  $d_r(x,y)=w_r(x-y)$. Moreover we define the shape of $(x,y)$:
\begin{equation*}
\shape(x,y):=(e_0,e_1,\dots,e_r) \text{ where } 
\begin{cases}
1\leq i\leq r, e_i:=\card\{j\ :\ w_r(x_j)=i\}\\
e_0:=n-(e_1+\dots+e_r).
\end{cases}
\end{equation*}
Another expression of $w_r(x)$ is $w_r(x)=\sum_{i} i e_i$.

If $B$ is the group of upper triangular matrices in $\Gl(\F_2^r)$, and
$B_{\text{aff}}$ the group of affine transformations of $\F_2^r$ combining the
translations by elements of $\F_2^r$ with $B$, the group
$G:=B_{\text{aff}}^n\rtimes S_n$
acts transitively on $X$. Since $B$ acting on $\F_2^r$ leaves $w_r$
invariant, it is clear that the action of $G$ on $X$ leaves the shape $\shape(x,y)$
invariant. More precisely, 
the orbits of $B$ on $\F_2^r$ are the sets $\{y\in \F_2^r :
w_r(x)=i\}$ and, consequently, the orbits of $G$ acting on $X^2$ are characterized by
the so-called shape of $(x,y)$. Since obviously
$\shape(x,y)=\shape(y,x)$ it is a symmetric space.
This space shares many common features with the Grassmann spaces,
especially from the point of view of the linear programming method
(see \cite{B1}, \cite{Barg}, \cite{MS}).

\subsubsection{\bf The space $X=\Ga$ under the action of $G=\Ga\times
  \Ga$:} the action of $G$ is by $(\gamma,\gamma')x=\gamma
x\gamma'^{-1}$. Then two pairs $(x,y)$ and $(x', y')$ are in the same
orbit under the action of $G$ iff $xy^{-1}$ and $x'y'^{-1}$ are in
the same conjugacy class of $\Ga$. 
Obviously $(x,y)$ and $(y^{-1}, x^{-1})$ are in the same $G$-orbit. We
are not quite in the case of a $G$-symmetric space however  the proof of
the commutativity of $\K$ of Proposition \ref{commute} remains valid because the variable change
$x\to x^{-1}$ leaves the Haar measure invariant. 

\subsection{Positive definite functions and Bochner
  theorem}\label{subsection Bochner pdf}

\begin{definition}\label{pdf}
A positive definite continuous function on $X$ is a function $F\in
\CC(X^2)$ such that $F(x,y)=\overline{F(y,x)}$ and one of the
following equivalent properties hold:
\begin{enumerate}
\item For all $n$, for all $(x_1,\dots, x_n)\in X^n$, for all
  $(\alpha_1,\dots,\alpha_n)\in \C^n$,
\begin{equation*}
\sum_{i,j=1}^n \alpha_i F(x_i,x_j)\overline{\alpha_j} \geq 0.
\end{equation*}
\item For all $\alpha\in \CC(X)$,
\begin{equation*}
\int_{X^2} \alpha(x)F(x,y)\overline{\alpha(y)} d\mu(x,y) \geq 0.
\end{equation*}
\end{enumerate}
This property will be denoted $F\succeq 0$.
\end{definition}

The first property means in other words that, for all choice of a
finite set of points $(x_1,\dots, x_n)\in X^n$, the matrix
$(F(x_i,x_j))_{1\leq i,j\leq n}$ is hermitian positive semidefinite.
The equivalence of the two properties results from compactness of
$X$. Note that, if $X$ is finite, $F$ is positive definite iff the
matrix indexed by $X$, with coefficients $F(x,y)$, is positive semidefinite.

We want to characterize those functions which are  $G$-invariant. This
characterization is provided by Bochner in \cite{Bo} in the case when
the space $X$ is $G$-homogeneous.
It is clear that the construction of previous subsection provides
positive definite functions. Indeed, 

\begin{lemma}\label{AE}
if $A\succeq 0$, then $\langle A, \overline{E}
\rangle$
is a $G$-invariant positive definite function.
\end{lemma}
\begin{proof} Let $\alpha(x)\in \CC(X)$. We compute
\begin{align*}
\int_{X^2} \alpha(x)\langle A, \overline{E}
\rangle \overline{\alpha(y)} d\mu(x,y) 
&=
\int_{X^2} \sum_{i,j=1}^m A_{i,j}
\alpha(x)E_{i,j}(x,y)\overline{\alpha(y)} d\mu(x,y) \\
&=
\sum_{i,j=1}^m A_{i,j} \int_{X^2}  
\alpha(x)E_{i,j}(x,y)\overline{\alpha(y)} d\mu(x,y) \\
&=
\sum_{i,j=1}^m \sum_{s=1}^d A_{i,j} \int_{X^2}  
\alpha(x)e_{i,s}(x)\overline{e_{j,s}(y)}\overline{\alpha(y)} d\mu(x,y) \\
&=
\sum_{i,j=1}^m \sum_{s=1}^d A_{i,j} \alpha_{i,s}\overline{\alpha_{j,s}}\\
&=
\sum_{s=1}^d \sum_{i,j=1}^m \alpha_{i,s} A_{i,j}
\overline{\alpha_{j,s}}\geq 0\\
\end{align*}
where $\alpha_{i,s}:=\int_X \alpha(x)e_{i,s}(x)d\mu(x)$.
\end{proof}

\begin{remark}\label{remark} The following properties are equivalent, for a $m
  \times m$ matrix  function $E(x,y)$:
\begin{enumerate}
\item For all $A\succeq 0$, $\langle A, \overline{E(x,y)} \rangle
  \succeq 0$
\item For all $(x_1,\dots, x_n)\in X^n$, $(\alpha_1,\dots,\alpha_n)\in
  \C^n$,
$\sum_{i,j} \alpha_i E(x_i,x_j)\overline{\alpha_j} \succeq 0$.
\end{enumerate}
The proof is left to the reader as an exercise (hint: use the fact that the cone of
positive semidefinite matrices is self dual).
\end{remark}

To start with, we extend the notations of the previous
subsection. We 
define matrices $E_k=E_k(x,y)$ associated to each isotypic component
$\I_k$, of size $m_k\times m_k$ (thus possibly of infinite size)
with coefficients $E_{k,i,j}(x,y)$ defined by:
\begin{equation*}
E_{k,i,j}(x,y):=\sum_{s=1}^{d_k} e_{k,i,s}(x)\overline{e_{k,j,s}(y)}.
\end{equation*}
If $F_k=(f_{k,i,j})_{1\leq i,j\leq m_k}$ is hermitian, and if
$\sum_{i,j}|f_{k,i,j}|^2<+\infty$, the sum
\begin{equation*}
\langle F_k, \overline{E_k} \rangle:=\sum_{i,j} f_{k,i,j}E_{k,i,j}
\end{equation*}
is $L^2$-convergent since the elements $e_{k,i,s}(x)\overline{e_{l,j,t}(y)}$ form a complete
system of orthonormal elements of $\CC(X^2)$. We say $F_k$ is
positive semidefinite ($F_k\succeq 0$) if $\sum_{i,j} \overline{\lambda_i} f_{k,i,j}
\lambda_j \geq 0$ for all $(\lambda_i)_{1\leq i\leq m_k}$ such that 
$\sum |\lambda_i|^2<+\infty$. Then, with the same proof as the one of
Lemma \ref{AE}, the function $\langle F_k, \overline{E_k} \rangle$ is
positive definite if $F_k\succeq 0$.
The following theorem provides a converse statement (see \cite{Bo}).

\begin{theorem}\label{Bochner pdf} $F\in\CC(X^2)$
is a $G$-invariant positive definite function if and
  only if 
\begin{equation}\label{Bochner pdf expression}
F(x,y)=\sum_{k\geq 0} \langle F_k, \overline{E_k(x,y)} \rangle
\end{equation}
where, for all $k\geq 0$, 
\begin{equation*}
F_k=\frac{1}{d_k\mu(X^2)}\int_{X^2} F(x,y)\overline{E_k(x,y)}d\mu(x,y)
\succeq 0, \end{equation*}
and the sum converges to $F$ for the $L^2$ topology. 
If moreover $G$ acts homogeneously on $X$, the sum 
\eqref{Bochner pdf expression} itself converges uniformly.

\end{theorem}

\begin{proof}
The elements $e_{k,i,s}(x)\overline{e_{l,j,t}(y)}$ form a complete
system of orthonormal elements of $\CC(X^2)$. Hence $F$ has a
decomposition 
\begin{equation*}
F(x,y)=\sum_{k,i,s,l,j,t} f_{k,i,s,l,j,t}e_{k,i,s}(x)\overline{e_{l,j,t}(y)}
\end{equation*}
where the convergence of the sum is $L^2$.
The condition $F(gx,gy)=F(x,y)$ translates to:
\begin{equation*}
f_{k,i,u,l,j,v}=\sum_{s,t} f_{k,i,s,l,j,t} R_{k,u,s}(g)\overline{R_{l,v,t}(g)}.
\end{equation*}
Integrating on $g\in G$ and applying the orthogonality relations of
Theorem \ref{orthog rel} shows that $f_{k,i,u,l,j,v}=0$ if $k\neq l$ or
$u\neq v$. Moreover it shows that $f_{k,i,u,k,j,u}$ does not depend
on $u$.  The resulting expression of $F$ reads:
\begin{equation*}
F(x,y)=\sum_{k\geq 0}\Big(\sum_{i,j} f_{k,i,j} E_{k,i,j}(x,y)\Big)
\end{equation*}
and 
\begin{equation*}
d_kf_{k,i,j}= \frac{1}{\mu(X^2)}\int_{X^2}
F(x,y)\overline{E_{k,i,j}(x,y)}d\mu(x,y),
\end{equation*}
which is the wanted expression, with $F_k:=(f_{k,i,j})_{1\leq i,j\leq m_k}$.

Now we show that $F_k\succeq 0$. Let, for $k,s$ fixed,
$\alpha(x)=\sum_{i} \alpha_i \overline{e_{k,i,s}(x)}$, with $\sum_i
|\alpha_i|^2<+\infty$. By density,  property (2) of Definition
\ref{pdf} holds for $\alpha\in L^2(X)$. We compute like in the proof
of Lemma \ref{AE}
\begin{equation*}
\int_{X^2} \alpha(x)F(x,y)\overline{\alpha(y)}
d\mu(x,y)=\sum_{i,j=1}^{m_k} \alpha_i f_{k,i,j}\overline{\alpha_j}
\end{equation*}
thus $F_k\succeq 0$.

In the case of $X$ being $G$-homogeneous, the uniform convergence of
the sum in \eqref{Bochner pdf expression} is proved in \cite{Bo}.

\end{proof}

In order to reduce linear programs involving $G$-invariant positive
definite functions to finite dimensional semidefinite programs, we
need to be able to approximate such functions uniformly  with finite sums of the
type \eqref{Bochner pdf expression}, in other words by functions built
form finite dimensional subspaces of $\CC(X)$. A necessary condition
is thus that all continuous functions on $X$ are uniformly
approximated by elements of some sequence of finite dimensional subspaces of
$\CC(X)$. Such subspaces are usually provided by the polynomial
functions of bounded degree, when it makes sense. More generally, let
us assume that there exists a sequence $(V_d)_{d\geq 0}$ of finite dimensional
$G$-subspaces of $\CC(X)$ such that $V_d\subset V_{d+1}$, and
$\cup_{d\geq 0} V_d$ is dense in $\CC(X)$ for the topology of uniform
convergence. For example, Peter-Weyl theorem provides such subspaces
when $X$ is $\Gamma$-homogeneous, for a compact group $\Gamma$
containing $G$. Then we have the following result:

\begin{theorem}\label{th uniform approx}
Under the above assumptions, if moreover $X$ is homogeneous under a
larger compact group $\Gamma$, and
if the irreducible  subspaces $H_{k,i}$ are chosen so that 
$H_{k,i}\subset V_d$ for all $1\leq i\leq m_{d,k}$ where $m_{d,k}$ is
the multiplicity of $R_k$ in $V_d$, then a $G$-invariant positive
definite function $F\in \CC(X^2)$
is the uniform limit of a sequence of positive definite functions
$F_d\in V_d\otimes V_d$ thus of the form
\begin{equation}\label{finite pdf}
F_d(x,y)=\sum_{k\geq 0} \langle F_{d,k}, \overline{E_k(x,y)} \rangle
\end{equation}
where $F_{d,k}$ is a matrix of size  $m_{d,k}$ (and thus the sum has a
finite number of non zero terms).
\end{theorem}

\begin{proof}
We proceed like in the proof of Peter Weyl theorem. Compact 
self-adjoint Hilbert-Schmidt
operators on $\CC(X^2)$ are of the form
\begin{equation*}
T_K(F)(x,y)=\int_{X^2} K((x,y),(z,t))F(z,t)d\mu(z,t).
\end{equation*}
We start to construct $K$ such $T_K(F)\succeq 0$ and $\| T_K(F)-F\|_{\infty}$ is
arbitrary small. The first condition is
fulfilled 
if $K$ can be expressed in the form $K((x,y),(z,t))=K_0(x,z)\overline{K_0(y,t)}$ where
$K_0(x,z)=\overline{K_0(z,x)}$. 
We take $\phi_0$ a
continuous function  on $\Gamma$; if
$\phi_0'$
denotes the left and right average of $\phi_0$ over $\Gamma_0$ (where
$X=\Gamma/\Gamma_0$), we take
$K_0(x,y)=\phi_0'(\gamma^{-1}\delta)$
for any $\gamma\in x$, $\delta\in y$).
Then with a suitable choice of $\phi_0$, $\| T_K(F)-F\|_{\infty}\leq
\epsilon$
(thanks to uniform continuity of $F$, it is enough that $\phi_0$ has
support contained in some prescribed open neighborhood of $1$, takes
values between $0$ and $1$, satisfies
$\phi_0(\gamma)=\phi_0(\gamma^{-1})$
and $\int_{\Gamma} \phi_0=|\Gamma_0|$). Moreover, $K_0$ is $\Gamma$-invariant.

We
can find $d\geq 0$ and $L_0(x,y)\in V_d\otimes V_d$ such that
$L_0(x,z)=\overline{L_0(z,x)}$ and  $\|L_0-K_0\|_{\infty}$
is arbitrary small. Replacing $L_0$ by its average on $G$ will not
change these three properties of $L_0$. Then, if
$L((x,y),(z,t)):=L_0(x,z)\overline{L_0(y,t)}$,
$T_L(F)$ comes arbitrary close to $T_K(F)$ for $\|\,\|_{\infty}$ and $T_L(F)\in V_d\otimes V_d$.
Now, $T_L(F)\succeq 0$, is invariant under $G$  and belongs to the finite dimensional space
$V_d\otimes V_d$ thus it has the announced form from Theorem
\ref{Bochner pdf}.

\end{proof}

Now the main deal is to compute explicitly the matrices $E_k(x,y)$ for
a given space $X$. The next section gives explicit examples of such
computation.

\section{Explicit computations of the matrices $E_k(x,y)$}

We keep the same notations as in previous section. Since the matrices
$E_k(x,y)$ are $G$-invariant, their coefficients are functions of the
orbits of $G$ acting on $X^2$. So the first task is to describe
these orbits. Let us assume that these orbits are parametrized by some
variables
$u=(u_i)$. Then we seek for explicit expressions of the form
\begin{equation*}
E_k(x,y)=Y_k(u(x,y)).
\end{equation*}

The measure $\mu$ induces a measure on the variables
that describe these orbits, for which the coefficients of $E_k$ are
pairwise orthogonal. This property of orthogonality turns to be very
useful, if not enough, to calculate the matrices $E_k$.

The easiest case is when the space $X$ is $2$-point
homogeneous for the action of $G$, because in this case the orbits of pairs
are parametrized by a single variable $t:=d(x,y)$. Moreover we have
already seen that in this case, the decomposition of $\CC(X)$ is multiplicity free
so the matrices
$E_k(x,y)$ have a single coefficient. 

\subsection{$2$-point homogeneous spaces.}

We summarize the results we have obtained so far:
\begin{equation*}
\CC(X)=\oplus_{k\geq 0} H_k
\end{equation*}
where $H_k$ are pairwise orthogonal $G$-irreducible subspaces; to each $H_k$ is associated a
continuous function $P_k(t)$ such that $E_k(x,y)=P_k(d(x,y))$ and 
\begin{equation*}
F\succeq 0 \Longleftrightarrow F=\sum_{k\geq 0} f_k P_k(d(x,y)) \text{
  with } f_k\geq 0.
\end{equation*}
$P_k(t)$ is called the zonal function associated to $H_k$. 
Since the subspaces $H_k$ are pairwise orthogonal, the functions
$P_k(t)$ are pairwise orthogonal for the induced measure. This property of
orthogonality is in general enough to determine them in a unique way.
We can also notice here that $P_k(0)=d_k$. This value is obtained with
the integration on $X$ of the formula $P_k(0)=\sum_{s=1}^{d_k} e_{k,1,s}(x)\overline{e_{k,1,s}(x)}$.

\subsection{$X=\{1,\dots,q\}$ under the action of $S_q$}\label{Sq}
This is a very easy case, which will play a role in the study of the
$q$-Hamming space. Since the constant
function $\1$ is $S_q$-invariant, we have the $S_q$ decomposition
$\CC(X)=\C\1\perp L$. Obviously, the action of $S_q$ on $X^2$ has
two orbits: the set of pairs  $(i,i)$, and the set of pairs
$(i,j)$ for $i\neq j$. Thus, from Proposition \ref{End2} and
Remark \ref{Rem}, $L$ is irreducible.
We let $z_0:=\1$ and choose  an orthonormal basis $(z_1,\dots,z_{q-1})$
of $L$. We want to compute the zonal function $E_L$ associated to $L$. We
have by definition $E_L(x,y)=\sum_{i=1}^{q-1}
z_i(x)\overline{z_i(y)}$ and $E_L$ takes only two different values:
one for $x=y$ and one for $x\neq y$. We have $E_L(0,0)=\dim(L)=q-1$ and
we can compute $E_L(0,1)$ easily using the fact that $E_L(0,y)$ is
orthogonal to $z_0$ thus $\sum_{y=1}^q E_L(0,y)=0=E_L(0,0)+(q-1)E_L(0,1)$.
Thus $E_L(0,1)=-1$.

\subsection{The $q$-Hamming space}\label{qHamming 1}
In the binary case we have already calculated the
functions $P_k(t)$ in \ref{Hamming 1}. Indeed, the irreducible subspaces $P_k$
afford the orthonormal basis $\{\chi_z, wt(z)=k\}$. So,
\begin{equation*}
E_k(x,y)=\sum_{wt(z)=k} \chi_z(x)\chi_z(y)=\sum_{wt(z)=k} (-1)^{z\cdot
  (x+y)}=K_k(d_H(x,y))
\end{equation*}
from \eqref{Krawtchouk}.
Now we treat the more general $q$-Hamming space. 
This is the space $H_{n,q}=F^n$ where $F$ is a finite set with $q$
elements
denoted $F=\{a_0,a_1,\dots,a_{q-1}\}$. The semidirect product $G=S_q^n
\rtimes S_n$ acts on $H_{n,q}$ and leaves the Hamming distance
invariant.  Here the permutation group $S_q$ acts on $F$ by
$\tau a_i=a_{\tau(i)}$ while the permutation group $S_n$ acts on
$H_{n,q}$ by $\sigma(x_1,\dots,x_n)=(x_{\sigma^{-1}(1)},\dots,
x_{\sigma^{-1}(n)})$. Moreover $G$ acts on $H_{n,q}$ $2$-point
homogeneously. 
The action of $S_q$ on $\CC(F)$ is studied in \ref{Sq} and we take the
same notations. 
We define $\phi=(\phi_1,\dots, \phi_n)\in \CC(H_{n,q})$ where
$\phi_i\in \{z_0,z_1,\dots,z_{q-1}\}$ by:
$\phi(x)=\prod_{i=1}^n \phi_i(x_i)$. These elements $\phi$  form an orthonormal system:
it is easy to see that
\begin{equation*}
\langle \phi,\psi\rangle=\prod_{i=1}^n \langle \phi_i,\psi_i\rangle.
\end{equation*}
We define the weight of $\phi$ by: $wt(\phi):= |\{1\leq i\leq n :
\phi_i\neq z_0\}|$. For $0\leq k\leq n$, let $P_k$ be the
subspace generated by the set of $\phi$ with $wt(\phi)=k$. The
dimension of $P_k$ is the number of such $\phi$, which is equal to
$(q-1)^k\binom{n}{k}$
and we have the decomposition 
\begin{equation}\label{dec Hammingq}
\CC(H_n)=P_0\perp P_1\perp\dots \perp P_n.
\end{equation}
An element $\tau\in S_q$ act trivially on $z_0$ and sends $z_i$
for $i\neq 0$ to a linear combination of
$z_1,\dots,z_{q-1}$. Thus for all $g\in G$, $g\phi$ is a linear combination of $\psi$'s with the
same weight as $\phi$ and  $G$ stabilizes $P_k$. The action
of $G$ on pairs of elements of $H_{n,q}$ has exactly $(n+1)$ orbits 
 corresponding to the $(n+1)$ values $0,1,\dots,n$ that the Hamming
 distance takes thus we can conclude that $P_k$ is irreducible from Proposition \ref{End2}.
Now we compute the zonal function $E_k(x,y)$ attached to $P_k$. By
definition we have
\begin{equation*}
E_k(x,y)=\sum_{\phi, wt(\phi)=k} \phi(x)\overline{\phi(y)}
\end{equation*}
and we want to calculate $P_k$ such that $P_k(t)=E_k(x,y)$ for any
$(x,y)$ with $d(x,y)=t$. We set $x=(a_1,\dots,a_1,a_0,\dots,a_0)$
where $t$ coordinates of $x$ are equal to $a_1$  and
$y=(a_0,\dots,a_0)$. For all $\phi$, we let $i:=|\{ j: 1\leq j\leq n : x_j=a_1 \text{
    and }\phi_j\neq z_0\}$ and reorder the set of $\phi\in P_k$
  according to $i$.
\begin{align*}
P_k(t)&=\sum_{i=0}^ k \binom{t}{i}\binom{n-t}{k-i} \sum_{j_1,\dots,
  j_k\neq 0} \prod_{u=1}^i
z_{j_u}(a_1)\overline{z_{j_u}(a_0)}\prod_{u=i+1}^k z_{j_u}(a_0)\overline{z_{j_u}(a_0)}\\
&=\sum_{i=0}^ k \binom{t}{i}\binom{n-t}{k-i} \Big(\sum_{s=1}^{q-1}
z_s(a_1)\overline{z_s(a_0)}\Big)^i \Big(\sum_{s=1}^{q-1}
z_s(a_0)\overline{z_s(a_0)}\Big)^{k-i}\\
&=\sum_{i=0}^ k \binom{t}{i}\binom{n-t}{k-i}
{E_{L}(a_1,a_0)}^i{E_{L}(a_0,a_0)}^{k-i}\\
&=\sum_{i=0}^ k \binom{t}{i}\binom{n-t}{k-i} (-1)^i(q-1)^{k-i}
\end{align*}
with the notations and results of \ref{Sq}. $P_k(t)$ is equal to the Krawtchouck
 polynomial $K_k^{n,q}(t)$ of parameters $q$ and $n$ which satisfies the following
characteristic properties:

\begin{enumerate}
\item $\deg(K_k^{n,q})=k$
\item $K_k^{n,q}(0)=(q-1)^k\binom{n}{k}$
\item Orthogonality relations: for all $0\leq k\leq l\leq n$
\begin{equation*}
\frac{1}{q^n}\sum_{w=0}^n \binom{n}{w}K_k^{n,q}(w)K_l^{n,q}(w)=\delta_{k,l}\binom{n}{k}(q-1)^k.
\end{equation*}
\end{enumerate}
The orthogonality relations are direct consequences of the
orthogonality of the subspaces $P_k$.

\subsubsection{\bf The Johnson space $J_n^w$:}\label{Johnson 2} with
the notations of subsection \ref{Johnson 1}, we have shown
the decomposition
\begin{equation*}
\CC(J_n^w)\simeq H_w\perp H_{w-1}\perp\dots \perp H_0
\end{equation*}
but not yet the irreducibility of $H_i$. So far their might by several
$P_{i,j}$, $j=1,\dots$ associated to $H_i$.
The zonal functions express as functions of $t:=|x\cap y|$ the
number of common ones in $x$ and $y$. The orthogonality relation is
easy to compute:
\begin{align*}
\sum_{x\in X} f(|x\cap y|)\overline{f'(|x\cap y|)}
&=\sum_{i=0}^n \card\{x\ :\ |x\cap y|=i\}f(i)\overline{f'(i)}\\
&=\sum_{i=0}^w \binom{w}{i}\binom{n-w}{w-i}f(i)\overline{f'(i)}\\
&=\sum_{i=0}^w \binom{w}{i}\binom{n-w}{i}f(w-i)\overline{f'(w-i)}.
\end{align*}
By induction on $k$ one proves that $P_{k,j}$ has degree at most $k$
in $t$. The conditions:
\begin{enumerate}
\item $\deg(Q_k)=k$
\item $Q_k(0)=1$
\item for all $0\leq k < l\leq n$
\begin{equation*}
\sum_{i=0}^w \binom{w}{i}\binom{n-w}{i}Q_k(i)Q_l(i)=0
\end{equation*}
\end{enumerate}
determine a unique sequence
$(Q_0,Q_1,\dots,Q_w)$. Thus there is only one $P_{k,j}$ for each $k$
and it is equal to $h_kQ_k(w-t)$. The polynomials $Q_k$ defined above
belong to the family of Hahn polynomials.

\subsubsection{\bf The sphere $S^{n-1}$:}\label{sphere 2} the distance on the sphere is
the angular distance $\theta(x,y)$. It appears more convenient to
express the functions in the variable $t=x\cdot y=\cos \theta(x,y)$. A standard
calculation shows that 
\begin{equation*}
\int_{S^{n-1}} f(x\cdot y) d\mu(y)=c_n\int_{-1}^1 f(t) (1-t^2)^{\frac{n-3}{2}} dt 
\end{equation*}
for some irrelevant constant $c_n$.
The conditions:
\begin{itemize}
\item $\deg(P_k^n)=k$
\item $P_k^n(1)=1$
\item For all $k\neq l$, $\int_{-1}^1
  P_k^n(t)P_l^n(t)(1-t^2)^{\frac{n-3}{2}} dt =0$
\end{itemize}
define a unique sequence of polynomials by standard arguments
(i.e. obtained by Gram Schmidt orthogonalization of the basis
$(1,t,\dots,t^k,\dots)$),
it is the sequence of so-called Gegenbauer polynomials with parameter
$n/2-1$ \cite{Sz}. The decomposition \ref{dec 1} of $\CC(S^{n-1})$ shows
that, to each $k\geq 0$ the function $P_k(x\cdot y)$ associated to
$H_k^n\simeq \Harm_k^n$ is polynomial in $x\cdot y$ and satisfies the
above conditions
except the normalization of $P_k(1)$ thus we have
$P_k(t)=h_k^n P_k^n(t)$.

\subsubsection{\bf Other $2$-point homogeneous
  spaces:}\label{2-homogeneous 2} as it is shown
in the above examples,  a sequence of orthogonal polynomials in one
variable is associated to each such space. In the case of the
projective spaces, it is a sequence of Jacobi polynomials. We refer to 
\cite{KL}, \cite{L}, \cite{Vil} for their determination in many cases
and for the applications to coding theory.

\subsection{Other symmetric spaces}
Now we turn to other cases of interest in coding theory, where the
space $X$ is
symmetric but not necessarily $2$-point homogeneous. Since the
decomposition of $\CC(X)$ is multiplicity free, the matrices
$E_k(x,y)$ still have a single coefficient which is a member of a
sequence of orthogonal polynomials, but this time multivariate. 
The first case ever studied (at least to my knowledge) is the case of
the non binary Johnson spaces \cite{TAG}, its associated functions are two
variables polynomials,
a mixture of Hahn and Eberlein polynomials.
We briefly discuss a few of these cases.

\subsubsection{\bf The Grassmann spaces:}\label{Grassmann 2} \cite{B1}
the orbits of $X^2$ are parametrized by the principal angles
$(\theta_1,\dots,\theta_m)$
(\ref{Grassmann 1}). The appropriate variables are the $y_i:=\cos^2\theta_i$. The
decomposition of $\CC(\Gmn)$ under $O(\R^n)$ (respectively $U(\C^n)$) together with the computation of the corresponding
sequence of orthogonal polynomials was performed in \cite{JC}. 
We focus here on the real case.
We recall that the irreducible representations of $\On$ are (up to a power of the
determinant)
naturally indexed by partitions
$\kappa=(\kappa_1,\dots,\kappa_n)$, where $\kappa_1\geq \dots\geq
\kappa_n\geq 0$ (we may omit the last parts if they are equal to $0$).
Following \cite{GW}, let them be denoted by $V_n^{\kappa}$.
For example, $V_n^{()}=\C\1$, and $V_n^{(k)}=\Harm_k$.

The length $\ell(\kappa)$ of a partition $\kappa$ is the
number of its non zero parts, and its degree 
$\deg(\kappa)$ also denoted by $|\kappa|$ equals $\sum_{i=1}^n\kappa_i$.

\smallskip
Then, the decomposition of $\CC(\Gmn)$ is as follows:

\begin{equation*}
\CC(\Gmn)\simeq \oplus V_n^{2\kappa}
\end{equation*}
where $\kappa$ runs over the partitions of length at most $m$ and
$2\kappa$ stands for partitions with even parts.
We denote by $\Pk(y_1,\dots,y_m)$ the zonal function associated to $V_n^{2\kappa}$.
It turns out that the $\Pk$ are symmetric polynomials in the
$m$ variables $y_1,\dots,y_m$, of degree $|\kappa|$,
with rational coefficients once they are normalized by the condition  $\Pk(1,\dots,1)~=~1$.
Moreover, the set $(\Pk)_{|\kappa|\leq k}$ is a basis of the space 
of symmetric polynomials in the variables $y_1,\dots,y_m$ of degree at most
equal to $k$, which is orthogonal for the induced
inner product calculated in \cite{JC},

$$d\mu=
\lambda \prod
_{\substack{i,j=1\\i<j}}^m|y_i-y_j|\prod_{i=1}^my_i^{-1/2}(1-y_i)^{n/2-m-1/2} dy_i
$$
(One recognizes a special case of the orthogonal measure associated to {\em generalized
  Jacobi polynomials} (\cite{La2}).

\subsubsection{\bf The ordered Hamming space:}\label{ordered Hamming 2} it follows from the
discussion in \ref{ordered Hamming 1} that the variables of the zonal
functions are the $(e_0,e_1,\dots, e_r)$. Elaborating on the computation
explained  above for the
Johnson space, one can see that in the case of finite
spaces, the weights of the induced measure are given by the number of
elements of the orbits of $X$ under the action of $\Stab(e)$ for
any $e\in X$. Taking $e=0^{rn}$, thus $\Stab(e)=B^n\rtimes S_n$, and
the orbit of $x$ is the set of elements with the same shape
$(f_0,\dots, f_r)$ as $x$. 
The number of such elements is $\binom{n}{f_0 \dots
  f_r}2^{\sum_i(i-1)e_i}$.
These are the weights associated to the multivariate Krawtchouk
polynomials. 

\subsubsection{\bf The space $X=\Ga$ under the action of $G=\Ga\times \Ga$:}
we need an explicit parametrization of the conjugacy classes of
$\Ga$, which is afforded by very few groups. Famous examples (if not
the only ones) are provided by the permutation groups and the unitary
groups.
In the first case the parametrization is by the decomposition in
disjoint cycles and in the second case it is by the eigenvalues. The
decomposition of $\CC(X)$ is given by Peter Weyl theorem 
\begin{equation*}
\CC(\Ga)=\sum_{R\in \RR} R\otimes R^*
\end{equation*}
and the associated functions $P_R(x,y)$ are the characters:
\begin{equation*}
P_R(x,y)=\chi_R(xy^{-1}).
\end{equation*}
In both cases ($S_n$ and $U(\C^n)$) the irreducible representations
are indexed by partitions $\lambda$ and there are explicit expressions 
for $P_{\lambda}$. In the case of the unitary group $P_{\lambda}(xy^{-1})$ are
  the so-called Schur polynomials evaluated at the eigenvalues of $xy^{-1}$.

\subsection{Three cases with non trivial multiplicities}

So far the computation of the matrices $E_k(x,y)$ in cases of non
trivial multiplicities has been worked out in very few cases. 
We shall discuss three  very similar
cases, namely the unit
sphere
of the Euclidean sphere (\cite{BV1}), the Hamming space (\cite{V1}), and the projective
geometry over $\F_q$ (\cite{BV4}), where the group considered is the stabilizer of one point.
In the case of the Hamming space, this computation amounts to the
computation of the Terwilliger algebra of the association scheme and was performed
initially by A. Schrijver in  \cite{Schrijver}, who treated also  the non binary Hamming
space \cite{GST}. The framework of group representations was used in
\cite{V1} to obtain the semidefinite matrices of \cite{Schrijver} in
terms of orthogonal polynomials. We present here the uniform treatment
of the Hamming space and of the projective geometry in the spirit of
\cite{Del2} adopted in \cite{BV4}.
We also generalize to the case of the stabilizer of many points in
the spherical case and enlighten the connection with the positive
definite functions calculated in \cite{M2}.

\subsubsection{\bf The unit sphere $S^{n-1}$, with
  $G:=\Stab(e,O(\R^n))$}\label{sphere 4} We continue the discussion initiated in
\ref{sphere 3} and we follow \cite{BV1}. 
Let $E_k^n(x,y)$ be the zonal matrix associated to the isotypic
subspace $\I_k$ related to $\Harm_k^{n-1}$   and to its decomposition
described in \ref{sphere 3}:
\begin{equation*}
\I_k=H^n_{k,k}\perp H^n_{k,k+1}\perp\dots
\end{equation*}
We index $E_k^n$ with $i,j\geq 0$ so that $E_{k,i,j}^n(x,y)$ is
related to the spaces $H^n_{k,k+i}$, $H^n_{k,k+j}$.
The orbits of $G$ on pairs of points $(x,y) \in X^2$
are characterized by the values of the three inner products
$u:=e\cdot x$, $v:=e\cdot y$ and $t:=x\cdot y$. Thus $(u,v,t)$ are the
variables of the zonal matrices and we let:
\begin{equation*}
E_k^n(x,y)=Y_k^n(u,v,t).
\end{equation*}

\begin{theorem}\label{Th Y}[\cite{BV1}] 
\begin{equation}\label{Y}
Y_{k,i,j}^n(u,v,t)= \lambda_{k,i}\lambda_{k,j}P_{i}^{n+2k}(u)P_{j}^{n+2k}(v)Q_k^{n-1}(u,v,t),
\end{equation}
where 
\[
\displaystyle Q_k^{n-1}(u,v,t):=\big((1-u^2)(1-v^2)\big)^{k/2}P_k^{n-1}\Big(\frac{t-uv}{\sqrt{(1-u^2)(1-v^2)}}\Big),
\]
and $\lambda_{k,i}$ are some real constants.
\end{theorem}

\begin{proof} We need an explicit construction of the spaces
  $H_{k,k+i}^{n-1}$. We refer to \cite[Ch.~9.8]{AAR}. For $x\in \Sn$,
let
\[
x=ue+\sqrt{1-u^2}\zeta,
\]
where $u=x\cdot e$ and $\zeta$ belongs to the unit sphere
$\Snm$ of $(\R e)^{\perp}$. With $f\in H_k^{n-1} \subset \CC(\Snm)$ we
associate $\varphi(f)\in \CC(\Sn)$ defined by:
\[
\varphi(f)(x)=(1-u^2)^{k/2}f(\zeta).
\]
Moreover, we recall that $H_k^n$ is a subspace of the space
$\Polk(\Sn)$ of polynomial functions in the coordinates of degree at
most $k$. 
Note that the multiplication by $(1-u^2)^{k/2}$ forces $\varphi(f)$ to
be a polynomial function in the coordinates of $x$. Clearly $\varphi$
commutes with the action of $G$. Hence $\varphi(H_k^{n-1})$ is a
subspace of $\Polk(\Sn)$ which is isomorphic to $\Harm_k^{n-1}$. It is
clear that these spaces are pairwise orthogonal.  More
generally, the set $\{\varphi(f)P(u) : f \in \Harm_{k}^{n-1}, \deg P
\leq i\}$ is a subspace of $\Polki(\Sn)$ which is isomorphic to $i+1$
copies of $\Harm_{k}^{n-1}$.  By induction on $k$ and $i$ there exist
polynomials $P_i(u)$ of degree $i$ such that
$H_{k,k+i}^{n-1}:=\varphi(H_k^{n-1})P_i(u)$ is a subspace of $H_{k+i}^n$. 
This construction proves the decomposition \eqref{dec res}.
Moreover, 
we can exploit the fact that the subspaces $H_{k,l}^{n-1}$ are
pairwise orthogonal to prove an orthogonality relation between the
polynomials $P_i$. Then this orthogonality relation will enable us to
identify the polynomials $P_i$ with  Gegenbauer polynomials, up to the
multiplication by a constant factor.
Let us recall that the measures on $\Sn$ and on $\Snm$ are related by:
\[
d\omega_n(x)=(1-u^2)^{(n-3)/2}du d\omega_{n-1}(\zeta).
\]
Whenever $i \neq j$ we have for all $f\in H_k^{n-1}$
\begin{align*}
0&=\frac{1}{\omega_n}\int_{\Sn} \varphi(f)P_i(u)\overline{\varphi(f)P_j(u)} d\omega_n(x)\\
&=\frac{1}{\omega_n}\int_{\Sn} |f(\zeta)|^2(1-u^2)^kP_i(u)\overline{P_j(u)} d\omega_n(x)\\
&=\frac{1}{\omega_n}\int_{\Snm} |f(\zeta)|^2d\omega_{n-1}(\zeta)\int_{-1}^1 (1-u^2)^{k+(n-3)/2}P_i(u)\overline{P_j(u)} du,
\end{align*}
from which we derive that 
\[
\int_{-1}^1 (1-u^2)^{k+(n-3)/2}P_i(u)\overline{P_j(u)}du=0;
\]
hence the polynomials $P_i(u)$ are proportional to $P_i^{n+2k}(u)$
(thus with real coefficients..).
We obtain an orthonormal basis of $H_{k,k+i}^{n-1}$ from an
orthonormal basis $(f_1,\ldots,f_h)$ of $H_k^{n-1}$ by taking
$e_{k,i,s}=\lambda_{k,i}\varphi(f_s)P_i^{n+2k}(u)$ for a suitable normalizing
factor $\lambda_{k,i}>0$. With these  basis we can compute
$E^n_{k,i,j}$:

\begin{eqnarray*}
& & E^n_{k,i,j}(x,y)=\sum_{s=1}^{h_k^{n-1}} e_{k,i,s}(x)\overline{e_{k,j,s}(y)}\\
&=&\sum_{s=1}^{h_k^{n-1}}  \lambda_{k,i}(1-u^2)^{k/2}f_s(\zeta)P_i^{n+2k}(u)\lambda_{k,j}(1-v^2)^{k/2}\overline{f_s(\xi)}P_j^{n+2k}(v)\\
&=&\lambda_{k,i}\lambda_{k,j} P_i^{n+2k}(u)P_j^{n+2k}(v)\big((1-u^2)(1-v^2)\big)^{k/2}\sum_{s=1}^{h_k^{n-1}} f_s(\zeta)\overline{f_s(\xi)}\\
&=& \lambda_{k,i}\lambda_{k,j}
  P_i^{n+2k}(u)P_j^{n+2k}(v)\big((1-u^2)(1-v^2)\big)^{k/2}h_k^{n-1}P_k^{n-1}(\zeta
  \cdot \xi),
\end{eqnarray*}
where we have written $y=ve+\sqrt{1-v^2}\xi$ and where the last
equality results from the analysis of zonal functions of $S^{n-1}$.
Since
\[
\zeta\cdot \xi=(t-uv
)/\sqrt{(1-u^2)(1-v^2)},
\]
we have completed the proof.
\end{proof}

\subsubsection{\bf The unit sphere $S^{n-1}$ with the action of
  $G:=\Stab(e_1,\dots,e_s,O(\R^n))$.}
We assume that $(e_1,\dots,e_s)$ is a set of orthonormal vectors. 
The group $G:=\Stab(e_1,\dots,e_s,O(\R^n))$ is isomorphic to
$O(\R^{n-s})$.
The orbit of a pair $(x,y)\in X^2$ under $G$ is characterized by the
data:
$t:=x\cdot y$, $u:=(x\cdot e_1,\dots,x\cdot e_s)$, 
$v:=(y\cdot e_1,\dots, y\cdot e_s)$. 
The decomposition \eqref{dec res} applied recursively shows that 
$\CC(S^{n-1})$ decomposes as the sum of $G$-irreducible subspaces
$H_{\underline{k}}$ where $\underline{k}=(k_0,\dots, k_{s})$,
$k_0\leq k_1\leq\dots \leq k_s$, with the properties:
\begin{equation*}
H_{\underline{k}}\subset H_{\underline{k}^{(r)}} \subset \operatorname{Pol}_{k_s},
\quad  H_{\underline{k}}\simeq \Harm_{k_0}^{n-s}
\end{equation*}
where $\underline{k}^{(r)}=(k_{s-r+1},\dots, k_s)$.
Thus, for a given $k_0$, the multiplicity of the isotypic  component $\I_{k_0}^d$
associated to $\Harm_{k_0}^{n-s}$ in $\Pold$ is the number of elements of
\begin{equation*}
K_d:=\{ (k_1,\dots, k_{s})\ : \  k_0\leq   k_1\leq\dots \leq k_s\leq d\}.
\end{equation*}
We construct the spaces $H_{\underline{k}}$ like in the proof of
Theorem \ref{Th Y}:
for $x\in S^{n-1}$, let
\begin{equation*}
x=u_1e_1+\dots+u_se_s+\sqrt{1-|u|^2}\zeta
\end{equation*}
where $u=(u_1,\dots,u_s)$ and $|u|^2=\sum_{i=1}^s u_i^2$.
Let $\varphi: H_{k_0}^{n-s}\to \CC(S^{n-1})$ be defined by
$\varphi(f)(x)=(1-|u|^2)^{k_0/2}f(\zeta)$.
Then $\varphi(H_{k_0}^{n-s})=H_{k_0^{s+1}}$ where
$k_0^{s+1}=(k_0,k_0,\dots,k_0)$
and we set, for $\underline{l}=(l_1,\dots,l_s)$,  
$H_{k_0,\underline{l}}:=u_1^{l_1}\dots u_s^{l_s}H_{k_0^{s+1}}$.
It is clear that $H_{k_0,\underline{l}}\simeq_G \Harm_{k_0}^{n-s}$ and
that $H_{k_0,\underline{l}}\subset \operatorname{Pol}_{d}$ if $l_1+\dots
+l_s\leq d-k_0$ thus, since
\begin{equation*}
K'_d:=\{l=(l_1,\dots,l_s) \ :\  l_i\geq 0,\  l_1+\dots+l_s\leq d-k_0\}
\end{equation*}
has the same number of elements as $K_d$,
\begin{equation*}
\I_{k_0}^d= \oplus_{\underline{l}\in K'_d} H_{k_0,l}.
\end{equation*}
This sum is not orthogonal but we can still use it to calculate
$E_{k_0}$, the change will be to $AE_k(x,y)A^*$ for some invertible
matrix $A$. The same calculation as in Theorem \ref{Th Y} shows that,
(up to a change to some $AY_kA^*$):
\begin{equation*}
Y_{k,
  \underline{i},\underline{j}}(u,v,t)=u^{\underline{i}-k}v^{\underline{j}-k}Q_k^{n-s}(u,v,t)
\end{equation*}
with the notations: $u^{\underline{i}-k}:=u_1^{i_1-k}u_2^{i_2-k}\dots
u_s^{i_s-k}$ and
\begin{equation*} 
Q_k^{n-s}(u,v,t)=\big((1-|u|^2)(1-|v|^2\big)^{k/2}P_k^{n-s}\Big(\frac{t-(u\cdot  v)}{\sqrt{(1-|u|^2)(1-|v|^2)}}\Big).
\end{equation*}
With Bochner Theorem \ref{Bochner pdf} we recover the description of the
multivariate positive definite functions on the sphere given in \cite{M2}.

\subsubsection{\bf The Hamming space and the projective
  geometry}\label{Hamming 3}
The
set of all $\F_q$-linear subspaces of $\F_q^n$, also called the
projective geometry,  is denoted by
$\PP(n,q)$. The linear group $\Gl(n,\F_q)$ acts on $\PP(n,q)$. The orbits of this
action are the subsets of subspaces of fixed dimension, i.e. the
$q$-Johnson spaces. 
If the Hamming space $\F_2^n$ is considered together with the  action
of the symmetric group $S_n$,
the orbits of this action are the Johnson spaces. In \cite{Del2} the
Johnson space and the $q$-Johnson spaces are treated in a uniform way
from the point of view of the linear programming method, the latter
being viewed as $q$-analogs of the former. Thus the Johnson space
corresponds to the value $q=1$. In particular the zonal
polynomials are computed and they turn to be $q$-Hahn polynomials.
Here we want to follow the same line for the determination of the
zonal matrices $E(x,y)$ in both cases.

We take the following notations: if $q$ is a power of a prime
number, we let $X= \PP(n,q)$ and $G=\Gl(n,\F_q)$, and, if $q=1$, we let $X$ be the Hamming space, identified with
the set of subsets of $\{1,\dots,n\}$,
and $G=S_n$ the symmetric group with its standard action on $X$.
Let
\begin{equation*}
|x|:=\Big\{
\begin{array}{ll}
wt(x) & \text{if } q=1\\
\dim(x) &\text{if } q>1
\end{array}
\end{equation*}
For all $w=0,\dots, n$, the space $X_w$ is defined by
\begin{equation*}
X_w=\{x\in X : |x|=w\}.
\end{equation*}
These subsets of $X$ are exactly the orbits of $G$.
The distance on $X$ is given in every case by the formula
\begin{equation}\label{distance}
d(x,y)=|x| +|y|-2|x\cap y|.
\end{equation}
The restriction of the distance $d$ to $X_w$ equals $d(x,y)=2(w-|x\cap
y|)$ and it is  a well known fact that $G$ acts 2-points homogeneously
on $X_w$. It is not difficult to see that the orbit of a pair $(x,y)$
under the action of $G$ is characterized by the triple $(|x|, |y|,
|x\cap y|)$.

Following the notations of \cite{Del2},
the $q$-binomial coefficient $\qbinom{n}{w}$ expresses the cardinality
of $X_w$. We have 

\begin{equation*}
\qbinom{n}{w}=\left\{
\begin{array}{ll}
\displaystyle \prod_{i=0}^{n-1} \frac{n-i}{w-i}=\binom{n}{w} & \text{if } q=1\\
\displaystyle\prod_{i=0}^{n-1} \frac{q^{n-i}-1}{q^{w-i}-1} &\text{if } q>1
\end{array}
\right.
\end{equation*}
In terms of the variable 
$$[x]=q^{1-x}\qbinom{x}{1}=\left\{
\begin{array}{ll}
x &\text{if } q=1\\
\displaystyle \frac{q^{-x}-1}{q^{-1}-1}&\text{if }q>1
\end{array}
\right.,$$
we have 
\begin{equation*}
\qbinom{n}{w}=q^{w(n-w)}\prod_{i=0}^{w-1} \frac{[n-i]}{[w-i]}=q^{w(n-w)}\frac{[n]!}{[w]![n-w]!}.
\end{equation*}

We have the obvious decomposition into pairwise orthogonal $G$-invariant subspaces:
\begin{equation*}
\CC(X)=\CC(X_0)\perp \CC(X_1)\perp\dots\perp \CC(X_n).
\end{equation*}
The decomposition of $\CC(X_w)$ into $G$-irreducible subspaces is
described in \cite{Del2}. We have 
\begin{equation*}
\CC(X_w)=H_{0,w}\perp H_{1,w}\perp\dots \perp H_{\min(w,n-w),w}
\end{equation*}
where the $H_{k,w}$ are pairwise  isomorphic for equal $k$ and different $w$. and 
pairwise non  isomorphic for different $k$.
The picture looks like:

\begin{equation*}
\begin{array}{cccccccc}
\CC(X)=&\CC(X_0) \perp &\CC(X_1)\perp &\dots  &\perp \CC(X_{\lfloor \frac{n}{2}\rfloor})\perp  &\dots &\perp \CC(X_{n-1})&\perp
\CC(X_n)\\
&&&&&&&\\
       & H_{0,0} \perp &H_{0,1} \perp      &\dots  &\perp H_{0,\lfloor         \frac{n}{2}\rfloor}\perp  &\dots &\perp H_{0,n-1}&\perp H_{0,n}\\
       &               & H_{1,1}\perp       &\dots  &             &     & \perp H_{1,n-1}&\\
&&&\ddots &\vdots &&&\\
&&&& H_{\lfloor \frac{n}{2}\rfloor,\lfloor \frac{n}{2}\rfloor}
\end{array}
\end{equation*}
where the columns represent the decomposition of $\CC(X_w)$ and 
the rows  the isotypic components of $\CC(X)$, i.e. 
the subspaces $\I_k:=H_{k,k}\perp H_{k,k+1}\perp\dots\perp H_{k,n-k}$, $0\leq k\leq
\lfloor \frac{n}{2}\rfloor$, with multiplicity $m_k=(n-2k+1)$.

Let, for all $(k,i)$ with $0\leq k\leq i\leq n-k$,
\begin{equation*}
\begin{array}{llll}
\psi_{k,i}: & \CC(X_k) & \to &\CC(X_i)\\
& f &\mapsto  &\psi_{k,i}(f) : 
\psi_{k,i}(f)(y)=\sum_{\substack{|x|=k\\x\subset y}}f(x)
\end{array}
\end{equation*}
and 
\begin{equation*}
\begin{array}{llll}
\delta_{k}: & \CC(X_k) & \to &\CC(X_{k-1})\\
& f &\mapsto  &\delta_{k}(f) : 
\delta_{k}(f)(z)=\sum_{\substack{|x|=k\\z\subset x}}f(x)
\end{array}
\end{equation*}
Obviously, these transformations commute with the action of $G$. The
spaces $H_{k,i}$ are defined by:
$H_{k,k}=\ker\delta_k$ and $H_{k,i}=\psi_{k,i}(H_{k,k})$. Moreover,

\begin{equation*}
h_k:=\dim(H_{k,k})= \qbinom{n}{k}-\qbinom{n}{k-1}.
\end{equation*}
We need later the following properties of $\psi_{k,i}$:
\begin{lemma}If $f,g\in H_{k,k}$,
\begin{equation}\label{e1}
\langle \psi_{k,i}(f),\psi_{k,i}(g)\rangle
=\qbinom{n-2k}{i-k}q^{k(i-k)}\langle f,g\rangle.
\end{equation}
Moreover, 
\begin{equation}\label{e2}
\psi_{i,j}\circ \psi_{k,i}=\qbinom{j-k}{i-k}\psi_{k,j}
\end{equation}
\end{lemma}

\begin{proof} \cite[Theorem 3]{Del2} proves \eqref{e1}. The relation
  \eqref{e2} is straightforward:
if $|z|=j$,
\begin{align*}
\psi_{i,j}(\psi_{k,i}(f))(z) &= \sum_{\substack{|y|=i\\y\subset z}}\psi_{k,i}(f)(y)
= \sum_{\substack{|y|=i\\y\subset z}}
\Big(\sum_{\substack{|x|=k\\x\subset y}} f(x) \Big)\\
&= \sum_{\substack{|x|=k\\x\subset z}}
\Big(\sum_{\substack{|y|=i\\x\subset y\subset  z}} 1\Big)f(x)
=\sum_{\substack{|x|=k\\x\subset z}}\qbinom{j-k}{i-k}f(x)\\
&=\qbinom{j-k}{i-k}\psi_{k,j}(f)(z).
\end{align*}
\end{proof}

Now we want to calculate the matrices $E_k$ of size $m_k=(n-2k+1)$
associated to each isotypic space $\I_k$. We fix an orthonormal basis
$(e_{k,k,1},\dots, e_{k,k,h_k})$ of $H_{k,k}$  and we define
$e_{k,i,s}:=\psi_{k,i}(e_{k,k,s})$. It is clear from the definitions
above that $e_{k,i,s}$ can be assumed to take real values.
From \eqref{e1}, for fixed $k$ and $i$, they form an orthogonal basis 
of $H_{k,i}$ with square norm equal to
$\qbinom{n-2k}{i-k}q^{k(i-k)}$. Normalizing them would
conjugate $E_{k}$ by a diagonal matrix, so we can omit to do it.
The matrix  $E_k$ is indexed with 
$i,j$ subject to  $k\leq i,j\leq n-k$. From the construction, we have
$E_{k,i,j}(x,y)=0$ if $|x|\neq i$ or $|y|\neq j$; 
since the matrix $E_k$ is zonal, we can define $P_{k,i,j}$ by
\begin{equation*}
E_{k,i,j}(x,y)=P_{k,i,j}(i-|x\cap y|)
\end{equation*}
and our goal is to calculate
the
 $P_{k,i,j}$. It turns out that these functions
express in terms of the so-called $q$-Hahn polynomials. 

We define the $q$-Hahn polynomials associated to the parameters
$n,i,j$ with $0\leq i\leq j\leq n$  to be the polynomials 
$Q_k(n,i,j; x)$ with $0\leq k\leq \min(i,n-j)$ uniquely determined by
the properties:
\begin{itemize}
\item $Q_k$ has degree $k$ in the variable $[x]$.
\item $(Q_k)_k$ is a sequence of polynomials orthogonal for the weights
\begin{equation*}
0\leq u\leq i \quad w(n,i,j; u)=\qbinom{i}{u}\qbinom{n-i}{j-i+u}q^{u(j-i+u)}
\end{equation*}
\item $Q_k(0)=1$
\end{itemize}
The polynomials $Q_k$ defined in \cite{Del2} and \ref{Johnson 2} correspond up to 
multiplication by $h_k$ to the parameters $(n,w,w)$ and,
with the notations of \cite{Du}, according to Theorem 2.5, again up to a multiplicative factor,
$Q_k(n,i,j;x)=E_m(i,n-i,j,i-x;q^{-1})$.
The combinatorial meaning of the above weights is the following:

\begin{lemma}\cite[Proposition 3.1]{Du}\label{l2}
Given $x\in X_i$, the number of elements $y\in X_j$ such that $|x\cap
y|=i-u$ is equal to $w(n,i,j;u)$.
\end{lemma}

\begin{theorem}\label{p1} If $k\leq i\leq j\leq n-k$, $|x|=i$, $|y|=j$,
\begin{equation*} 
E_{k,i,j}(x,y)= |X| h_k\frac{\qbinom{j-k}{i-k}\qbinom{n-2k}{j-k}}{\qbinom{n}{j}\qbinom{j}{i}}q^{k(j-k)}Q_k(n,i,j; i-|x\cap y|)
\end{equation*}
If $|x|\neq i$ or $|y|\neq j$, $E_{k,i,j}(x,y)=0$.
\end{theorem}

\begin{proof} We proceed in two steps: the first step \eqref{e3} calculates
  $P_{k,i,j}(0)$  and the second step \eqref{e4}  obtains the orthogonality
  relations. 

\begin{lemma} With the above notations,
\begin{equation}\label{e3}
P_{k,i,j}(0)=|X|h_k
\frac{\qbinom{j-k}{i-k}\qbinom{n-2k}{j-k}}{\qbinom{n}{j}\qbinom{j}{i}}q^{k(j-k)}.
\end{equation}
\end{lemma}

\begin{proof} We have $P_{k,i,j}(0)=E_{k,i,j}(x,y)$ for all $x,y$ with
  $|x|=i$, $|y|=j$, $x\subset y$. Hence
\begin{align*}
P_{k,i,j}(0) &= \frac{1}{\qbinom{n}{j}\qbinom{j}{i}}
\sum_{\substack{|x|=i, |y|=j\\x\subset y}} E_{k,i,j}(x,y)\\
&= \frac{1}{\qbinom{n}{j}\qbinom{j}{i}}
\sum_{\substack{|x|=i, |y|=j\\x\subset y}} 
\sum_{s=1}^{h_k} e_{k,i,s}(x)e_{k,j,s}(y)\\
&= \frac{1}{\qbinom{n}{j}\qbinom{j}{i}} 
\sum_{s=1}^{h_k}\sum_{|y|=j} \Big(\sum_{\substack{|x|=i \\x\subset y}}  e_{k,i,s}(x)\Big)e_{k,j,s}(y)\\
&= \frac{1}{\qbinom{n}{j}\qbinom{j}{i}} 
\sum_{s=1}^{h_k}\sum_{|y|=j} \psi_{i,j}(e_{k,i,s})(y) e_{k,j,s}(y)\\
\end{align*}
Since, from \eqref{e2}
$$ \psi_{i,j}(e_{k,i,s})=\psi_{i,j}\circ
\psi_{k,i}(e_{k,k,s})=\qbinom{j-k}{i-k}\psi_{k,j}(e_{k,k,s})
=\qbinom{j-k}{i-k}e_{k,j,s},$$
we obtain 
\begin{align*}
P_{k,i,j}(0) &= 
\frac{1}{\qbinom{n}{j}\qbinom{j}{i}} 
\sum_{s=1}^{h_k}\sum_{|y|=j} \qbinom{j-k}{i-k}e_{k,j,s}(y) e_{k,j,s}(y)\\
&= 
\frac{\qbinom{j-k}{i-k}}{\qbinom{n}{j}\qbinom{j}{i}}  
\sum_{s=1}^{h_k} |X|\langle e_{k,j,s}, e_{k,j,s} \rangle
=
|X|h_k\frac{\qbinom{j-k}{i-k}\qbinom{n-2k}{j-k}}{\qbinom{n}{j}\qbinom{j}{i}}q^{k(j-k)}
\end{align*}
from \eqref{e1}.
\end{proof}

\begin{lemma} With the above notations,
\begin{equation}\label{e4}
\sum_{u=0}^i w(n,i,j;u)P_{k,i,j}(u)P_{l,i,j}(u)= \delta_{k,l}|X|^2h_k
\frac{\qbinom{n-2k}{i-k}
\qbinom{n-2k}{j-k}q^{k(i+j-2k)}}{\qbinom{n}{i}}.
\end{equation}
\end{lemma}

\begin{proof} We compute $\Sigma:=\sum_{y\in X} E_{k,i,j}(x,y)
  E_{l,i',j'}(y,z)$.

\begin{align*}
\Sigma &= 
\sum_{y\in X} \sum_{s=1}^{h_k}\sum_{t=1}^{h_l}e_{k,i,s}(x)e_{k,j,s}(y)e_{l,i',t}(y)e_{l,j',t}(z)\\
&= \sum_{s=1}^{h_k}\sum_{t=1}^{h_l}e_{k,i,s}(x)e_{l,j',t}(z)\Big(\sum_{y\in  X}e_{k,j,s}(y)e_{l,i',t}(y)\Big)\\
&= \sum_{s=1}^{h_k}\sum_{t=1}^{h_l}e_{k,i,s}(x)e_{l,j',t}(z)|X|\langle
e_{k,j,s},e_{l,i',t}\rangle\\
&= \sum_{s=1}^{h_k}\sum_{t=1}^{h_l}
e_{k,i,s}(x)e_{l,j',t}(z)|X|\qbinom{n-2k}{j-k}q^{k(j-k)}\delta_{k,l}\delta_{j,i'}\delta_{s,t}\\
&=
\delta_{k,l}\delta_{j,i'}|X|\qbinom{n-2k}{j-k}q^{k(j-k)}\sum_{s=1}^{h_k}
e_{k,i,s}(x)e_{l,j',s}(z)\\
&= \delta_{k,l}\delta_{j,i'}|X|\qbinom{n-2k}{j-k}q^{k(j-k)}E_{k,i,j'}(x,z).
\end{align*}
We obtain, with $j=i'$, $j'=i$, $x=z\in X_i$, taking account of
$E_{l,j,i}(y,x)=E_{l,i,j}(x,y)$, 

\begin{equation*}
\sum_{y\in X_j} E_{k,i,j}(x,y)E_{l,i,j}(x,y)
=\delta_{k,l}|X|\qbinom{n-2k}{j-k}q^{k(j-k)} E_{k,i,i}(x,x).
\end{equation*}
The above identity becomes in terms of $P_{k,i,j}$
\begin{equation*}
\sum_{y\in X_j} P_{k,i,j}(i-|x\cap y|)P_{l,i,j}(i-|x\cap y|)
=\delta_{k,l}|X|\qbinom{n-2k}{j-k}q^{k(j-k)} P_{k,i,i}(0).
\end{equation*}
Taking account of \eqref{e3} and Lemma \ref{l2}, we 
obtain \eqref{e4}.
\end{proof}

To finish the proof of Proposition \ref{p1}, it remains to prove
that $P_{k,i,j}$ is a polynomial of degree at most $k$ in the
variable $[u]=[|x\cap y|]$. It follows from the reasons invoked
in \cite{Del2} in the case $i=j$ (see the proof of Theorem 5).
\end{proof}

\begin{remark}\label{remark2}
In the case $q=1$, i.e. the Hamming space, we could have followed the
same line as for the sphere in order to decompose $\CC(H_n)$ under the
action of $G$. We could have started from the
decomposition
of $\CC(H_n)$ \eqref{dec Hamming} under the action of
$\Gamma:=T\rtimes S_n=\Aut(H_n)$ and then we could  have decomposed each space $P_k$
under the action of $G=\Stab(0^n, \Gamma)$. But we have a
$G$-isomorphism from $\CC(X_w)=\CC(J_n^w)$ to $P_w$  given by:
\begin{align*}
\CC(J_n^w) &\to P_w \\
f &\mapsto \sum_{wt(y)=w} f(y)\chi_y
\end{align*}
Note that the inverse isomorphism is the Fourier transform on
$(\Z/2\Z)^n$. So we  pass from  one to the other decomposition of
$\CC(H_n)$ through Fourier transform.
\end{remark}

\section{An SDP upper bound for codes from positive definite functions}

In this section we want to explain how the computation of
the continuous $G$-invariant positive definite functions
on $X$ can be used for applications to coding theory.
In coding theory, it is of great importance to estimate the maximal
number of elements of a finite subset  $C$ of a space $X$, where $C$
is submitted to some constraints. Typically $X$ is a metric space
with $G$-invariant distance $d(x,y)$ and the constraints are related
to the values taken by the distance on pairs of elements of $C$. In
the following we concentrate on the basic case where the requirement
is that the distance takes non zero values at least equal to some minimum $\delta$.
We denote by $D$ the set of all values taken by $d(x,y)$ and we define $D_{\geq \delta}=D\cap [\delta,+\infty[$ and
\begin{equation*}
A(X,\delta):=\max\{ \card(C) \ :\ d(c,c')\geq \delta \text{ for all
}c\neq c', (c,c')\in C^2\}.
\end{equation*}
We first focus on an upper bound for $A(X,\delta)$, which is
obtained very obviously from the optimal value of the following program:

\begin{definition}
\begin{equation}\label{mXdelta}
\begin{array}{lll} 
m(X,\delta) & = \inf\big\{ \quad t : &  F\in \CC(X^2),\ \overline{F}=F,\ F \succeq 0 \\
&&    F(x,x)\leq t-1,\\
&&   F(x,y)\leq -1 \quad  d(x,y)\geq \delta \big\}\\
\end{array}
\end{equation}
\end{definition}
Then we obtain an upper bound for $A(X,\delta)$:

\begin{theorem}\label{bound}
\begin{equation*}
A(X,\delta)\leq m(X,\delta).
\end{equation*}
\end{theorem}
\begin{proof}
For a feasible solution  $F$, and for $C\subset X$ with $d(C)\geq
\delta$ we have 
\begin{equation*}
0\leq \sum_{(c,c')\in C^2} F(c,c') \leq (t-1)|C| -|C|(|C|-1)
\end{equation*}
thus $|C|\leq t$.
\end{proof}

Now the group $G$ comes into play. From a feasible solution $F$ one
can construct a $G$-invariant feasible solution $F'$ with
the same objective value:
\begin{equation*}
F'(x,y)=\int_G F(gx,gy) dg
\end{equation*}
thus we can add to the conditions defining the feasible solutions of
$m(X,\delta)$
that $F$ is $G$-invariant. Then we can apply Bochner characterization
of the $G$-invariant positive definite functions (Theorem \ref{Bochner
  pdf}). Moreover we have also seen in Theorem \ref{th uniform approx} that
if $X$ is a homogeneous space, the finite sums of type \eqref{finite pdf}
are arbitrary close for $\|\|_{\infty}$ to the $G$-invariant positive definite functions
on $X$, so we can replace $F$ by an expression of the form
\eqref{finite pdf}
in the SDP $m(X,\delta)$. Moreover, we replace $E_k(x,y)$ with its
expression $Y_k(u(x,y))$ in terms of the orbits of pairs and we take
account of the fact that $\overline{F}=F$.
All together, with the notations of subsection \ref{subsection Bochner pdf}  we obtain the (finite) semidefinite programs:

\begin{equation}\label{mXdelta-d}
\begin{array}{lll} 
m^{(d)}(X,\delta) & = \inf\big\{ \quad t : &  F_0\succeq 0,\dots,
F_k\succeq 0 ,\dots\\
&&   \sum_{k\geq 0} \langle F_k, \tilde{Y}_k(u(x,x)) \leq t-1,\\
&&  \sum_{k\geq 0} \langle F_k, \tilde{Y}_k(u(x,y))  \leq -1 \quad  d(x,y)\geq \delta \big\}\\
\end{array}
\end{equation}
where the matrices $F_k$ are real symmetric, with size $m_{d,k}$, 
and $\tilde{Y}_k(u(x,y))=Y_k(u(x,y))+\overline{Y_k(u(x,y))}$. We
insist that in the above program only a finite number of integers $k$
are to be taken account of because $m_{d,k}\neq 0$ for a finite number of integers $k$.
Thus we have $m(X,\delta)\leq m^{(d)}(X,\delta)$ and 
\begin{equation*}
\lim_{d\to+\infty}m^{(d)}(X,\delta) =m(X,\delta).
\end{equation*}

\subsection{The $2$-point homogeneous spaces}\label{2-homogeneous 3}
We recall that  a sequence of orthogonal functions $(P_k)_{k\geq 0}$
is associated to $X$ such that the $G$-invariant positive definite
functions 
have the expressions
\begin{equation*}
F(x,y)=\sum_{k\geq 0} f_k P_k(d(x,y)) \text{ with } f_k\geq 0.
\end{equation*}
Then 
\begin{equation*}
\begin{array}{lll}
& m(X,\delta)=\inf \ \{\ 1+\sum_{k\geq 1} f_k \ : \ & f_k\geq 0,\\
        &&1+\sum_{k\geq 1} f_kP_k(i) \leq 0 \text{ for all }i\in
D_{\geq \delta}\ 
 \}
\end{array}
\end{equation*}

We restate Theorem \ref{bound}
 in the classical form of Delsarte linear programming bound:

\begin{theorem}\label{lp-2}
Let $F(t)=f_0+f_1P_1(t)+\dots +f_dP_d(t)$. If $f_k\geq 0$ for all
$0\leq k\leq d$ and $f_0>0$, and if $F(t)\leq 0$ for all $t\in D_{\geq
  \delta}$,
then 
\begin{equation*}
A(X,\delta)\leq \frac{f_0+f_1+\dots +f_d}{f_0}.
\end{equation*}
\end{theorem}

\noindent {\bf Example:} $X=S^7$, $d(x,y)=\theta(x,y)$, $d(C)=\pi/3$. This
  value of the minimal angle corresponds to the kissing number
  problem.
A very good kissing configuration is well known: it is the root system
$E_8$, also equal to the set of minimal vectors of the $E_8$
lattice. It has $240$ elements and the inner products take the values
$\pm 1$, $0$, $\pm 1/2$. We recall that the zonal polynomials associated to
the unit sphere are proportional to the Gegenbauer polynomials
$P_k^n$ in the variable $x\cdot y$. If $P(t)$ obtains the tight bound
$240$ in Theorem \ref{lp-2}, then we must have $P(t)\leq 0$ for $t\in
[-1,1/2]$ and $P(-1)=P(\pm 1/2)=P(0)=0$ (as part of the {\it complementary
slackness conditions}).
The simplest possibility is $P=(t-1/2)t^2(t+1/2)^2(t+1)$. One can
check that 
\begin{equation*}
\frac{320}{3}P=P^8_0 +\frac{16}{7}P^8_1 +\frac{200}{63}P^8_2
+\frac{832}{231}P^8_3 +\frac{1216}{429}P^8_4+\frac{5120}{3003}P^8_5
+\frac{2560}{4641}P^8_6
\end{equation*}
and that 
\begin{equation*}
\frac{P(1)}{f_0}=240.
\end{equation*}

Thus the kissing number in dimension $8$ is equal to $240$. This
famous proof
is due independently to Levenshtein \cite{Le} and Odlysko and Sloane \cite{OS}. A proof of
uniqueness  derives from the analysis of this bound (\cite{BS}).
For the kissing number problem, this miracle reproduces only for
dimension $24$ with the set of shortest  vectors of the Leech
lattice. For the other similar cases in $2$-point homogeneous spaces
we refer to \cite{L}.

\smallskip
It is not always possible to apply the above ``guess of a good
polynomial'' method.
In order to obtain a more systematic way to apply Theorem \ref{lp-2},
one can of course restrict the degrees of the  polynomials
to some reasonable value, but needs also to overcome
the problem that the conditions $F(t)\leq 0$ for $t\in [-1,1/2]$
represent infinitely many linear inequalities. One possibility is to
sample the
interval and then a posteriori study the extrema of the approximated optimal solution
found by an algorithm that solves the  linear program with finitely
many unknowns and inequalities. It is the method adopted in
\cite{OS}, where upper bounds for the kissing number in dimension $n\leq
30$ have been computed. We want to point out that polynomial
optimization methods using SDP give another way to handle this
problem. A polynomial $Q(t)\in \R[t]$ is said to be a sum of squares
if $Q=\sum_{i=1}^r Q_i^2$ for some $Q_i\in \R[t]$. Being a sum of
squares is a SDP condition since it amounts to ask that 
\begin{equation*}
Q=(1,t,\dots, t^k) F (1,t,\dots, t^k)^*\text{ with } F\succeq 0.
\end{equation*} 
Here $k$ is an upper bound for the degrees of the polynomials $Q_i$.
Now we can relax the condition that $F(t)\leq 0$ for $t\in [-1,1/2]$
to
$F(t)=-Q(t)-Q'(t)(t+1)(t-1/2)$ with $Q$ and $Q'$ being sums of
squares.
A theorem of Putinar claims that in fact the two conditions are
equivalent (but the degree of the polynomials
under the squares are unknown).

\smallskip
A very nice achievement of the linear programming method in $2$-point
homogeneous spaces is the derivation of an asymptotic upper bound for the rate of codes 
(i.e. for the quotient $\log \card(C)/\dim(X)$) obtained from 
the so-called Christoffel-Darboux kernels. This method was first
discovered for the Hamming and Johnson spaces \cite{MRRW} and then
generalized to the unit sphere \cite{KL} and to all other $2$-point
homogeneous spaces
\cite{L}. It happens to be the best known upper bound for the
asymptotic range. In \cite{KL} an asymptotic bound is derived for the
density of sphere packings in Euclidean space which is also the best known.

\subsection{Symmetric spaces} 
For these spaces, which are not $2$-point homogeneous, there may be
several distance functions of interest  which are $G$-invariant. For example, the
analysis of performance of codes in the Grassmann spaces for the MIMO
channel
\cite{Creignou} involves both the
chordal distance:
\begin{equation*}
d_c(p,q):=\sqrt{\sum_{i=1}^m \sin^2 \theta_i(p,q)}
\end{equation*}
and the product pseudo distance (it is not a distance in the metric sense):
\begin{equation*}
d_p(p,q):=\prod_{i=1}^m \sin \theta_i(p,q).
\end{equation*}
The reformulation of Theorem \ref{bound} leads to a theorem
of the type \ref{lp-2} for any symmetric function
of the $y_i:=\cos\theta_i$ with the Jacobi polynomials
$P_{\mu}(y_1,\dots,y_m)$ instead of the $P_k$.
For a general symmetric space, a theorem
of the type \ref{lp-2}  is obtained, where the sequence of polynomials
$P_k(t)$ is replaced by a sequence of multivariate polynomials, and
the set $D_{\delta}$ is replaced by some compact subspace of the
domain of 
the variables of the zonal functions, i.e. of the orbits of $G$ acting
on pairs.
Then one can derive
explicit upper bounds, see \cite{T} for the permutation codes,
\cite{B1}
for the real Grassmann codes, \cite{Roy2} and \cite{Creignou} for the
complex Grassmann codes, \cite{CD} for the unitary codes,
\cite{Barg} and 
\cite{MS} for the ordered codes. Moreover an asymptotic bound is
derived in \cite{B1} and \cite{Barg}.

\subsection{Other spaces with true SDP bounds}
An example where the bound \eqref{mXdelta} does not boil down to an LP
is provided by the spaces $\PP(n,q)$ endowed with the distance
\eqref{distance}
for which the matrices $E_k$ are computed in section \ref{Hamming 3}
(see \cite{BV4}). 
In this case the group $G$ is the largest group that acts on the SDP.

Indeed, it is useless to restrict the symmetrization of the program
\eqref{mXdelta} to some subgroup of the largest group $G$ that preserves
$(X,d)$.
However, another interesting possibility is to change the restricted
condition
$d(x,y)\geq \delta$ in $A(X,\delta)$ for the conditions:
\begin{equation}\label{cap}
d(x,y)\geq \delta, \ d(x,e)\leq r, \ d(y,e)\leq r
\end{equation}
where $e\in X$ is a fixed point. Then the new $A(X,e,r,\delta)$ is the
maximal number of elements of a code with minimal distance $\delta$ in the ball
$B(e,r)\subset X$. Here the group that leaves the program invariant is
$\Stab(e,G)$. The corresponding bounds for codes in spherical caps
where computed in \cite{BV3} using the expressions of the zonal
matrices of \ref{sphere 4}.

\medskip
We end this section with some comments on these SDP bounds. We have
indeed generalized the framework of the classical LP bounds but the
degree of understanding of the newly defined bounds is far from the
one of the classical LP bounds after the work done since \cite{Del2},
see e.g. \cite{L}.
It would be very interesting to have a better understanding of the
best functions $F$ that give the best bounds, to analyse explicit
bounds and to analyse the asymptotic range, although partial results
in these directions have already been obtained. The fact that one has
to deal with multivariate polynomials introduces great difficulties
when one tries to follow the same lines as for the classical one variable cases.
A typical example is provided by the configuration of $183$ points on
the half sphere that seems numerically to be an optimal configuration
for the one sided kissing number, and for which we failed to find the
proper function $F$ leading to a tight bound (see \cite{BV4}).

\section{Lov\'asz theta}

In this section we want to establish a link between the program
\eqref{mXdelta} and the so-called Lov\'asz theta number. This number
  was introduced by Lov\'asz in the seminal
paper \cite{Lovasz} in order to compute the capacity of the
pentagon. This remarkable result is the first of a long list of
applications. This number is the optimal  solution of a semidefinite
program, thus is ``easy to calculate'', and offers an approximation of
invariants of graphs that are ``hard to calculate''. Since then many
other SDP relaxations of hard problems have been proposed 
in graph theory and in other domains.

\subsection{Introduction to Lov\'asz theta number}

A graph $\Gamma=(V,E)$ is a finite set $V$ of vertices together with a
finite set $E$ of edges, i.e. $E\subset V^2$. An independence set $S$ is a
subset of $V$ such that $S^2\cap E=\emptyset$. The independence number
$\alpha(\Gamma)$ is the maximum of the number of elements of an independence
set.
It is a hard problem to determine the independence number of a graph. The
connection with coding theory is as follows: a code $C$ of a finite
space $X$ with minimal
distance $d(C)\geq \delta$ is an independence set of the graph $\Gamma(X,\delta)$ 
which vertex set is equal to $X$ and which edge set is equal to
$E_{\delta}:=\{(x,y)\in X^2\ :\ d(x,y)\in ]0,\delta[\}$.
Thus the determination of $A(X,\delta)$ is the same as the
determination of the independence number of this graph.

Among the many definitions of Lov\'asz theta, we choose one which
generalizes nicely to infinite graphs. For $S\subset V$, let $\1_S$ be
the characteristic function of $S$. Let
\begin{equation*}
M(x,y):=\frac{1}{|S|} \1_S(x)\1_S(y).
\end{equation*}
The following properties hold for $M$:

\begin{enumerate}
\item $M\in \R^{n\times n}$, where $|V|=n$, and $M$ is symmetric
\item $M\succeq 0$
\item $\sum_{x\in V} M(x,x)=1$
\item $M(x,y)=0$ if $(x,y)\in E$
\item $\sum_{(x,y)\in V^2} M(x,y)=|S|$.
\end{enumerate}

\begin{definition}
The theta number of the graph $\Gamma=(V,E)$ with $V=\{1,2,\dots,n\}$  is 
\begin{equation}\label{theta primal}
\begin{array}{lll} 
  \vartheta(\Gamma) & = \max\big\{  \sum_{i,j} B_{i,j} : &
  B\in \R^{n\times n},\ B \succeq 0 \\
&&   \sum_i B_{i,i}=1,\\
&&   B_{i,j}=0 \quad  (i,j)\in E\big\}\\
\end{array}
\end{equation}
\end{definition}

The dual program for $\vartheta$ has the same optimal value and is
equal to: 
\begin{equation}\label{theta dual}
\begin{array}{lll} 
\vartheta(\Gamma) & = \min\big\{ \quad t : &  B \succeq 0 \\
&&    B_{i,i}=t-1,\\
&&   B_{i,j}=-1 \quad  (i,j)\notin E\big\}\\
\end{array}
\end{equation}
The complementary graph of $\Gamma$ is denoted
$\overline{\Gamma}$. The chromatic number $\chi(\Gamma)$ is the
minimum number of colors needed to color the vertices so that no two
connected vertices receive the same color. In other words it is a
minimal partition of the vertex set with  independence sets.
Then the so-called Sandwich theorem holds:
\begin{theorem}
\begin{equation*}
\alpha(\Gamma)\leq \vartheta(\Gamma)\leq \chi(\overline{\Gamma})
\end{equation*}
\end{theorem}
\begin{proof} The discussion prior to the theorem proves the first
  inequality.
For the second inequality, let $c:V \to \{1,\dots,k\}$ be a coloring
of $\overline{\Gamma}$. Then the matrix $C$ with $C_{i,j}=-1$ if
$c(i)\neq c(j)$, $C_{i,i}=k-1$ and $C_{i,j}=0$ otherwise provides a
feasible solution of \eqref{theta dual}.
\end{proof}

\subsection{Symmetrization and the $q$-gones}
Now we assume that $G$ is (a subgroup of) the automorphism group
$\Aut(\Gamma)$ of the graph. Then, $G$ acts also on the above defined
semidefinite programs. Averaging on $G$ allows to construct a
$G$-invariant optimal feasible solution $B'$ from any optimal feasible
solution $B$ with the same objective value:
\begin{equation*}
B'_{i,j}:=\frac{1}{|G|}\sum_{g\in G} B_{g(i),g(j)}.
\end{equation*}
Thus one can restrict in the above programs to the $G$-invariant
matrices. Then one can exploit the method developed in previous
sections, in order to obtain a description of the $G$-invariant
$B\succeq 0$ form the decomposition of the space $\CC(V)$ under the
action of $G$. We illustrate the method in the case of the $q$-gone $C_q$.
There we have $V=G=\Z_q$ the group of integers modulo  $q$. Let $\zeta_q$
be a fixed primitive root of $1$ in $\C$. Let $\chi_k:\Z_q\to \C^*$ be
defined  
by $\chi_k(x)=\zeta_q^{kx}$. The characters
of $\Z_q$ are the $\chi_k$ for $0\leq k\leq q-1$ and we have the
decomposition
\begin{equation*}
\CC(\Z_q)=\oplus_{k=0}^{q-1} \C \chi_k.
\end{equation*}
According to Theorem \ref{Bochner pdf}, the $G$-invariant positive definite functions on $V$ are
exactly the functions $F(x,y)$ of the form:
\begin{equation*}
F(x,y)=\sum_{k=0}^{q-1} f_k
\chi_k(x)\overline{\chi_k(y)}=\sum_{k=0}^{q-1} f_k \zeta_q^{k(x-y)}
\end{equation*}
with $f_k\geq 0$. The ones taking real values have the form
\begin{equation*}
F(x,y)=\sum_{k=0}^{\lfloor q/2\rfloor} f_k \cos((x-y)2k\pi/q),\quad f_k\geq 0.
\end{equation*}
When one replaces in $\vartheta$ the expression $B_{i,j}=F(i,j)$, the
SDP transforms into a LP on the variables $f_k$. More
precisely,
we compute $\sum_{(x,y)\in V^2} F(x,y)=q^2f_0 $ and $\sum_{x\in V}
F(x,x)=q\sum_k f_k$. Thus we obtain (after a change of $qf_k$ to
$f_k$):
\begin{equation*}
\begin{array}{lll} 
  \vartheta(C_q) & = \max\big\{  qf_0 : &
  f_k\geq 0,\ 0\leq k\leq \lfloor q/2\rfloor, \\
&&  \displaystyle \sum_{k=0}^{\lfloor q/2\rfloor} f_k=1,\\
&&  \displaystyle \sum_{k=0}^{\lfloor q/2\rfloor} f_k\cos(2k\pi/q)=0\\
\end{array}
\end{equation*}
The optimal value of this very simple linear program, is obtained for 
$f_1=f_2=\dots=f_{\lfloor q/2\rfloor-1}=0$, and equals
\begin{equation*}
\vartheta(\C_q)=\left\{ \begin{array}{ll}
\frac{q}{2} &\text{ if } q \text{ is even }\\
\frac{q\cos (\pi/q)}{1+\cos(\pi/q)} &\text{ if }q \text{ is odd }.
\end{array}
\right.
\end{equation*}
Note that when $q$ is even, the independence number of the $q$-gone is
exactly $q/2$. If the independence number of a graph as simple as the
$q$-gone is not a great deal (it is of course equal to $\lfloor
q/2\rfloor$), a more challenging issue is to determine its capacity. In
general, the capacity $C(\Ga)$ of a graph $\Ga$ is defined to be
\begin{equation*}
C(\Ga)=\lim_{n\to+\infty} {\alpha(\Ga^n)}^{1/n}.
\end{equation*}
Here the graph  $\Ga^n$ is defined as follows:  its vertex set is equal
to $V^n$ and an edge connects $(x_1,\dots,x_n)$ and $(y_1,\dots,y_n)$
iff
for all $1\leq i\leq n$ either $x_i=y_i$ or $(x_i,y_i)\in E$. 
Introduced by Shannon in 1956, this number represents the effective
size of an alphabet used to transmit information through the channel associated to the graph $\Ga$ (where
two symbols are undistinguable if they are connected by an edge). If
the capacity of a graph is in general very difficult to calculate, the
theta number of a graph provides an upper bound for it 
because $\vartheta(\Ga^n)= \vartheta(\Ga)^n$ (see \cite{Lovasz}). This upper
bound is an equality for the pentagon since on one hand $\vartheta(C_5)=\sqrt{5}$
from our previous computation, and on the other hand it is easy to see that
$\alpha((C_5)^2)=5$ (while $\alpha(C_5)=2$); this is the way taken by
Lov\'asz in \cite{Lovasz} to prove that $C(C_5)=\sqrt{5}$. 
The determination of the capacity of the $q$-gone for $q$ odd and
greater than $5$ is still opened.

\subsection{Relation with Delsarte bound and with $m(X,\delta)$}

We introduce a slightly stronger bound for $\alpha(\Gamma)$ with $\vartheta'$ and its dual form:
\begin{equation}\label{theta prime primal}
\begin{array}{lll} 
  \vartheta'(\Gamma) & = \max\big\{  \sum_{i,j} B_{i,j} : &  B \succeq
  0,\ B\geq 0 \\
&&   \sum_i B_{i,i}=1,\\
&&   B_{i,j}=0 \quad  (i,j)\in E\big\}
\end{array}
\end{equation}

\begin{equation}\label{theta prime dual}
\begin{array}{lll} 
\vartheta'(\Gamma) & = \min\big\{ \quad t : &  B \succeq 0 \\
&&    B_{i,i}\,\leq\, t-1,\\
&&   B_{i,j}\,\leq \,-1 \quad  (i,j)\notin E\big\}\\
\end{array}
\end{equation}
Since $M(x,y)\geq 0$, we still have that $\alpha(\Gamma)\leq
\vartheta'(\Gamma)$.
Again one can restrict in the above programs to the $G$-invariant
matrices. It was recognized independently by 
McEliece, Rodemich, Rumsey, and Schrijver \cite{Schrijver2} that
Delsarte bound of Theorem \ref{lp-2} for $A(H_n,\delta)$ is equal to $\vartheta'$ for the
graph $\Gamma(X,\delta)$, once the feasible set is restricted to the
$\Aut(H_n)$-invariant matrices, and similarly for the other finite $2$-point
homogeneous spaces. Indeed, by virtue of Theorem \ref{Bochner pdf}, the
matrices $B$ turn to be of the form
$B(x,y)=\sum_{k\geq 0} f_k P_k(d(x,y))$. This symmetrization process
is of great importance, not only because it has the great advantage to
change an SDP to an LP, but also because it does change the
complexity of the problem. Indeed, there are algorithms with
polynomial complexity  that do compute approximations of the optimal
value of SDP's, thus algorithms with polynomial complexity {\em in the
  number of vertices} of $\Gamma$ for $\vartheta$. But the graphs
arising from coding theory have in general an exponential number of
vertices, e.g. $2^n$ for the Hamming graph. It is important to insist
that the symmetrized theta has polynomial complexity in $n$.

Now we can see that the program $m(X,\delta)$ \eqref{mXdelta} is a natural
generalization of $\vartheta'$ for  metric spaces under the assumptions of
Section \ref{Harmonic analysis}. We refer to \cite{BNOV}
for a more general discussion  about generalized theta where also
chromatic numbers are involved.

\section{Strengthening the LP bound for binary codes}

In this section we explain how the zonal matrices $E_k(x,y)$ related
to the binary Hamming space computed
in \ref{Hamming 3} are exploited in \cite{Schrijver} in order to
strengthen the LP bound. We shall work with the primal programs so we
start to recall the primal version of \eqref{mXdelta} in the case of
the Hamming space.

We recall that  the sequence of orthogonal functions $(P_k)_{0\leq
  k\leq n}$
with $P_k=K_k$ the Krawtchouk polynomials 
is associated to $H_n$ such that $P_k(d(x,y))\succeq 0$. As a
consequence, we have for all $k\geq 0$
\begin{equation*}
\sum_{(c,c')\in C^2} P_k(d(c,c'))\geq 0.
\end{equation*}
We introduce the variables $x_i$, for $i\in [0\dots n]$
\begin{equation}\label{var}
x_i:=\frac{1}{\card(C)}\card\{(c,c')\in C^2 \ :\ d(c,c')=i\}.
\end{equation}
They satisfy the properties:
\begin{enumerate}
\item $x_0=1$
\item $x_i\geq 0$
\item $\sum_{i} x_iP_k(i)\geq 0$ for all $k\geq 0$
\item $x_i=0$ if $i\in [1\dots \delta-1]$
\item $\card(C)=\sum_{i} x_i$.
\end{enumerate}
With these properties which are linear inequalities, we obtain the following
linear program which is indeed the dual of \eqref{mXdelta}:

\begin{equation*}
\begin{array}{lll}
& \sup\  \{\  1+\sum_{i=\delta}^n x_i \ : \ & x_i\geq 0,\\
        &&1+\sum_{i=\delta}^n x_iP_k(i)\geq 0\text{  for all } 1\leq
  k\leq n\ \}\\
\end{array}
\end{equation*}
where we have taken into account $P_0=1$.

We recall that to every $0\leq k\leq
\lfloor \frac{n}{2}\rfloor$,   we have associated a  matrix
$E_k(x,y)\succeq 0$ of size $n-2k+1$. In particular, for all
$C\subset H_n$ (see the remark \ref{remark}),
\begin{equation*}
\sum_{(c,c')\in C^2} E_k(c,c')\succeq 0.
\end{equation*}
These constraints are not interesting for pairs because they are not
stronger than the linear inequalities coming from the Krawtchouk polynomials. They are
only interesting if triples of points are involved: namely we associate
to
$(x,y,z)\in H_n^3$ the matrices 
\begin{equation*}
F_k(x,y,z):=E_k(x-z,y-z).
\end{equation*}
We have for all
$C\subset H_n$, and for all $z\in H_n$,
\begin{equation*}
\sum_{(c,c')\in C^2} F_k(c,c',z)\succeq 0
\end{equation*}
which leads to the two positive semidefinite conditions:
\begin{equation}\label{eq1}
\left\{
\begin{array}{ll}
&\sum_{(c,c',c'')\in C^3} F_k(c,c',c'')\succeq 0\\
&\sum_{(c,c')\in C^2, \ c''\notin  C} F_k(c,c',c'')\succeq 0
\end{array}
\right.
\end{equation}
Theorem  \ref{p1}, expresses the coefficients of $E_k(x-z,y-z)$ in
terms of $wt(x-z)$, $wt(y-z)$, $wt(x-y)$; so with 
$a:=d(y,z)$, $b:=d(x,z)$,
$c:=d(x,y)$, we have for some matrices $T_k(a,b,c)$,
\begin{equation*}
F_k(x,y,z)=T_k(a,b,c).
\end{equation*}
We introduce the unknowns $x_{a,b,c}$ of the SDP.
Let
\begin{equation*}
\Omega:=\Big\{(a,b,c) \in [0\dots n]^3 : 
\begin{array}{ll}
& a+b+c\equiv 0\mod 2\\
&a+b+c\leq 2n\\
&c\leq a+b\\
&b\leq a+c\\
&a\leq b+c
\end{array}\Big\}
\end{equation*}
It is easy to check that $\Omega=\{(d(y,z), d(x,z),
d(x,y))\ :\ (x,y,z)\in H_n^3\}$. Let, for $(a,b,c)\in \Omega$, 
\begin{equation*}
x_{a,b,c}:=\frac{1}{\card(C)} \card\{(x,y,z)\in C^3:
d(y,z)=a, d(x,z)=b, d(x,y)=c\}.
\end{equation*}
Note that 
\begin{equation*}
x_{0,c,c}=\frac{1}{\card(C)} \card\{(x,y)\in C^2:
d(x,y)=c\}
\end{equation*}
thus the old  variables $x_i$  \eqref{var} of the linear program are part
of these new variables.
We need a last notation: let 
\begin{equation*}
\begin{array}{ll}
t(a,b,c)&:=\card\{z \in H_n : d(x,z)=b \text{ and }d(y,z)=a\}\text{ for
} d(x,y)=c\\
&=\binom{c}{i}\binom{n-c}{a-i}
\text{ where } a-b+c=2i 
\end{array}
\end{equation*}
Then, if $C$ is a binary code with minimal distance at least equal to
$\delta$,
the following inequalities hold for $x_{a,b,c}$ :
\begin{enumerate}
\item $x_{0,0,0}=1$
\item $x_{a,b,c}\geq 0$
\item $x_{a,b,c}=x_{\tau(a),\tau(b),\tau(c)}$ for all permutation
  $\tau$ of $\{a,b,c\}$
\item $x_{a,b,c}\leq t(a,b,c)x_{0,c,c}$, $x_{a,b,c}\leq
  t(b,c,a)x_{0,a,a}$, $x_{a,b,c}\leq t(c,a,b)x_{0,b,b}$.
\item $\sum_{a,b,c}  T_k(a,b,c)x_{a,b,c}\succeq 0$ for all $0\leq
  k\leq \lfloor \frac{n}{2}\rfloor$
\item $\sum_{a,b,c} T_k(a,b,c)(t(a,b,c)x_{0,c,c}- x_{a,b,c}) \succeq
  0$
for all $0\leq
  k\leq \lfloor \frac{n}{2}\rfloor$
\item $x_{a,b,c}=0$ \text{ if } $a$, $b$ or $c\in ]0,\delta[$.
\item $\card(C)=\sum_{c} x_{0,c,c}$.
\end{enumerate}
Conditions (5) and (6) are equivalent to \eqref{eq1}.
Condition (7) translates the assumption that $d(C)\geq \delta$.
Thus an upper bound on $\card(C)$ is obtained with the optimal value
of the program that maximizes $\sum_{c} x_{0,c,c}$ under the
constraints
(1) to (7). 
This upper bound is at least as good as the LP bound because the
SDP program does contain the LP program of \ref{2-homogeneous
  3}. Indeed,
the sum of the two SDP conditions \eqref{eq1} is equivalent to
\begin{equation*}
\sum_{z\in H_n} E_k(x-z,y-z)\succeq 0.
\end{equation*}
We claim that this set of conditions when $k=0,1,\dots,
\lfloor\frac{n}{2} \rfloor $  is equivalent to the set of conditions
$P_k(d(x,y))\succeq 0$
for $k=0,\dots,n$. 
Indeed let $B_k(x,y):=\sum_{z\in H_n} E_k(x-z,y-z)$. Up to a change of
$B_k(x,y)$ to $AB_k(x,y)A^*$, we assume that $E_k$ was constructed
using the decomposition of $\CC(H_n)$ first under $\Gamma:=T\rtimes S_n=\Aut(H_n)$ then under
$G$ (see Remark \ref{remark2}). Clearly $B_k$ is
$\Gamma$-invariant. Since $x\to E_{k,i,j}(x,y)\in P_i$ and $P_i$ is a
$\Gamma$-module,
also $x\to B_{k,i,j}(x,y)\in P_i$ and similarly $y\to
\overline{B_{k,i,j}(x,y)}\in P_j$. But $P_i$ and $P_j$ are non
isomorphic $\Ga$-modules for $i\neq j$ thus $B_{k,i,j}(x,y)=0$ for
$i\neq j$. Since $P_i$ is $\Ga$-irreducible,
$B_{k,i,i}(x,y)=\lambda_iP_i(d(x,y))$ for some $\lambda_i>0$ that can
be computed with $B_k(x,x)$.
So we have proved that the linear program associated to $H_n$ like in
\ref{2-homogeneous 3} is contained in the SDP program obtained from
the above conditions (1) to (7).
Moreover it turns out that in some explicit cases of small dimension 
 the SDP bound is strictly better than the LP bound (see
\cite{Schrijver}).

\smallskip
A similar strengthening of the LP bound for the Johnson space and for
the spaces of non binary codes where obtained in \cite{Schrijver} and \cite
{GST}. In the case of the spherical codes, for the same reasons as for
the LP bound, one has to deal with the dual program, see  \cite{BV1}.

\end{document}